\documentclass[12pt]{article}
\usepackage{graphicx}
\usepackage{setspace}
 \DeclareGraphicsExtensions{.eps, .ps}
\usepackage{soul,color}
\usepackage{amsmath}
\usepackage{amsthm}
\usepackage{amsfonts}
\usepackage{newfloat}
\usepackage{natbib}
\usepackage{hyperref}

\usepackage{float}

\DeclareFloatingEnvironment[name={Supplementary Figure}]{suppfigure}

\usepackage{setspace}
\usepackage[margin=1in]{geometry}
\theoremstyle{plain}
\newtheorem{theorem}{Theorem}

\newtheorem{remark}{Remark}

\newtheorem{lemma}{Lemma}

\newtheorem{condition}{Condition}
\newtheorem{theorem*}{Theorem}

\usepackage{authblk}
\usepackage{graphicx} 
\usepackage{grffile}
\usepackage{caption}
\usepackage{subcaption}

\def\neyman{\textsc{n}}
\def\ding{\textsc{dfm}}

\def\deri{\textnormal{d}}
\usepackage{bbm}
\def\I{\mathbbm{1}}
\def\bs{}
\usepackage[textsize=tiny, textwidth = 2cm, shadow]{todonotes}
\usepackage{color}

\def\T{{ \mathrm{\scriptscriptstyle T} }}

\def\pr{\text{pr}}
\def\osl{{\color{black}\alpha}}

\begin{document}
\title{Power and Sample Size Calculations for Rerandomization}
\author{ZACH BRANSON}
\affil{Department of Statistics, Carnegie Mellon University, Pittsburgh, PA 15213, U.S.A., zach@stat.cmu.edu}

\author{XINRAN LI}
\affil{Department of Statistics, University of Illinois, Champaign, IL 61820, U.S.A., xinranli@illinois.edu}

\author{PENG DING}
\affil{Department of Statistics, University of California, Berkeley, CA 94720, U.S.A., pengdingpku@berkeley.edu}

\maketitle

\begin{abstract}
Power analyses are an important aspect of experimental design, because they help determine how experiments are implemented in practice. It is common to specify a desired level of power and compute the sample size necessary to obtain that power. Such calculations are well-known for completely randomized experiments, but there can be many benefits to using other experimental designs. For example, it has recently been established that rerandomization, where subjects are randomized until covariate balance is obtained, increases the precision of causal effect estimators. This work establishes the power of rerandomized treatment-control experiments, thereby allowing for sample size calculators. We find the surprising result that, while power is often greater under rerandomization than complete randomization, the opposite can occur for very small treatment effects. The reason is that inference under rerandomization can be relatively more conservative, in the sense that it can have a lower type-I error at the same nominal significance level, and this additional conservativeness adversely affects power.
This surprising result is due to treatment effect heterogeneity, a quantity often ignored in power analyses. We find that heterogeneity increases power for large effect sizes but decreases power for small effect sizes.
\end{abstract}

\section{Introduction}

We consider two-arm randomized experiments, with the aim of estimating causal effects. Randomized experiments are frequently considered the gold standard of causal inference because even simple estimators, such as the average difference in outcomes between groups, are unbiased if subjects are completely randomized between the two groups 
\citep{Neyman1923, imbens2015causal}. However, it is often beneficial to use covariate information when randomizing subjects. For example, estimators are usually more precise if subjects are grouped into similar blocks and randomized within  blocks \citep{fisher1936design, box1978statistics,imai2008variance, miratrix2013adjusting, tabord2018stratification, pashley2021insights, bai2022optimality}. More generally, one can use rerandomization \citep{morgan2012rerandomization}, where subjects are randomized until covariate balance is obtained. Block randomization is a special case, where subjects are randomized until the blocking variable is balanced. However, it's unclear how to form blocks when there are many, possibly continuous, covariates \citep{bruhn2009pursuit}. Meanwhile, rerandomization can naturally incorporate categorical and continuous covariates in the covariate balance criterion.

Since \cite{morgan2012rerandomization}, many works have established the benefits of rerandomization; this includes experiments with tiers of covariates \citep{morgan2015rerandomization}, sequential experiments \citep{zhou2018sequential}, factorial experiments \citep{branson2016improving,li2020rerandomizationFactorial}, stratified experiments \citep{wang2021rerandomization}, experiments with clusters \citep{lu2022design}, and experiments with high-dimensional covariates \citep{branson211ridge,zhang2021pca}. A common theme is that causal effect estimators are more precise under rerandomization than complete randomization as long as covariates are associated with outcomes. In particular, \cite{li2018asymptotic} established that, asymptotically, the mean-difference estimator has narrower symmetric quantile ranges under rerandomization than under complete randomization. Thus, confidence intervals for average causal effects are narrower under rerandomization.

Intuitively, because rerandomization increases estimation precision, we would expect it to also increase testing power. While this has been alluded to in the literature \citep{morgan2012rerandomization}, the power of rerandomized experiments has not been established. Power analyses are an important aspect of experimental design, because they help determine how experiments are conducted in practice. For example, when designing an experiment, it is common to specify a desired level of power, and then compute the sample size necessary to obtain that level of power \citep{lenth2001some,maxwell2008sample,chow2017sample}. There are many publicly available sample size calculators for completely randomized experiments.

This work establishes testing power under rerandomization, thereby allowing for sample size calculators. We focus on the mean-difference estimator, such that we can leverage results from \cite{li2018asymptotic}. We establish the surprising result that, while power is often greater under rerandomization than complete randomization, the opposite can occur when the average treatment effect is small. The main reason is that inference under rerandomization can be more conservative when there is treatment effect heterogeneity.
More precisely, variance estimators for individual effects overestimate the true variance by the same amount under both designs, which results in a larger proportion of overestimation under rerandomization. 
Specifically, at the same nominal level, 
testing under rerandomization can have a smaller actual type-I error than that under complete randomization.
This additional conservativeness has an adverse effect on power that can outweigh the precision benefits of rerandomization, but only for small treatment effects.

To compare power and sample size between complete randomization and rerandomization, we establish a dispersive ordering between their respective Normal and non-Normal distributions. Our results also quantify how power and sample size are affected by treatment effect heterogeneity, which is often ignored in power analyses. More generally, this work adds to the literature on power analyses for complex experiments, such as two-stage randomized experiments \citep{jiang2020statistical}, regression discontinuity designs \citep{schochet2009statistical}, and difference-in-difference designs \citep{schochet2021statistical}. In the supplementary material, we illustrate the sample size gains practitioners would see from rerandomization under various scenarios, and also show how to implement power and sample size calculations in our R package \texttt{rerandPower}.

\section{Notation for Treatment-Control Experiments}\label{s:notation}

Consider a treatment-control experiment with $N$ subjects indexed by $i=1,\dots,N$. Let $Z = (Z_1,\dots,Z_N)^\T$ denote the binary group indicator, where $Z_i = 1$ denotes treatment and $Z_i = 0$ denotes control, and define $X \equiv (X_1,\dots,X_N)^\T$ as the $N \times K$ covariate matrix. Finally, let $Y_i(1)$ and $Y_i(0)$ denote the potential outcomes for subject $i$, where $Y_i(1)$ denotes the outcome subject $i$ yields if assigned to $Z_i = 1$, and $Y_i(0)$ is analogously defined. Throughout, we assume $Y_i(1)$ and $Y_i(0)$ are fixed; observed outcomes are random only to the extent that $Z$ is random.

Because only $Y_i(1)$ or $Y_i(0)$ is observed for each subject, none of the individual treatment effects $\tau_i = Y_i(1) - Y_i(0)$ are fully observed. Nonetheless, average treatment effects, and other estimands, can still be estimated. We assume the goal is to well-estimate the average treatment effect, defined as 
\begin{align*}
	\tau &= N^{-1} \sum_{i=1}^N \tau_i = \bar{Y}(1) - \bar{Y}(0).
\end{align*}
There are many possible estimators for $\tau$; for simplicity, we focus on the mean-difference estimator
\begin{align*}
	\hat{\tau} = 
	N_1^{-1} \sum_{i=1}^N Z_i Y_i(1) - N_0^{-1} \sum_{i=1}^N (1 - Z_i)Y_i(0),
\end{align*}
where $N_1$ and $N_0$ denote the number of treated and control subjects, respectively.
We will study testing power based on $\hat{\tau}$ under rerandomization. Power will depend on $\tau$, the variance of potential outcomes, and the variance of individual treatment effects, which are respectively defined as:
\begin{align*}
		S_z^2 &= (N - 1)^{-1} \sum_{i=1}^N \left\{ Y_i(z) - \bar{Y}(z) \right\}^2
		\text{ for }z=0,1, 
		\quad 
		S^2_{\tau} = (N-1)^{-1} \sum_{i=1}^N \left( \tau_i - \tau \right)^2. 
\end{align*}

\section{Power and Sample Size Under Rerandomization} \label{s:rerandomizedExperiments}

\subsection{Inference and Power} \label{ss:rerandInferencePower}

In completely randomized experiments, a fixed number of $N_1$ subjects are assigned to treatment and a fixed number of $N_0$ subjects are assigned to control, completely at random (\citealt[Chapter 4]{imbens2015causal}). Alternatively, the covariates $X$ can inform the design, which will affect inference and power. We consider the rerandomization scheme of \cite{morgan2012rerandomization}, where subjects are completely randomized to treatment until $M \leq a$ for a prespecified threshold $a$, where $M$ is the Mahalanobis distance:
\begin{align}
	M = \frac{N_1 N_0}{N} \left( \bar{X}_1 - \bar{X}_0 \right)^{\T} (S^2_{X})^{-1} \left( \bar{X}_1 - \bar{X}_0 \right) \label{eqn:md},
\end{align}
where $\bar{X}_1$ and $\bar{X}_0$ are $K$-length vectors of covariate means in the treatment and control groups, respectively, and $S^2_{X} \equiv (N-1)^{-1} \sum_{i=1}^N (X_i - \bar{X}) (X_i - \bar{X})^\T$ denotes the covariance matrix of $X$. 

Under complete randomization, the mean-difference estimator $\hat{\tau}$ is unbiased and asymptotically Normally distributed \citep{Neyman1923,li2017general}:
\begin{equation}
\begin{aligned}
	&V^{-1/2} N^{1/2} (\hat{\tau}-\tau) \sim \mathcal{N}(0,1), \hspace{0.1 in} \text{ where } \hspace{0.1in} V = p_1^{-1}S_1^2 + p_0^{-1}S_0^2 - S^2_\tau, \label{eqn:compRandDistribution} 
\end{aligned}
\end{equation}
where $p_1=N_1/N$ and $p_0=N_0/N$.
Meanwhile, under rerandomization, $\hat{\tau} \mid M \leq a$ follows a non-Normal distribution that depends on
covariates' association with potential outcomes. Although the covariates and potential outcomes are fixed, a linear projection of potential outcomes on covariates can account for some of the potential outcomes' variance, and thus covariates can still have an association with the potential outcomes.
To quantify this association, we define the following quantities for $z \in \{0,1\}$:
\begin{align*}
    S^2_{z \mid X} &= S_{z, X} (S^2_{X})^{-1} S_{ X, z}, \quad \text{where} \quad S_{z, X} =  S^\T_{X, z} = (N-1)^{-1} \sum_{i=1}^N \{Y_i(z) - \bar{Y}(z)\}(X_i - \bar{X})^\T, \\
    S^2_{\tau \mid X} &= S_{\tau, X} (S^2_{X})^{-1} S_{ X, \tau}, \quad \text{where} \quad S_{\tau, X} =  S^\T_{X, \tau} = (N-1)^{-1} \sum_{i=1}^N (\tau_i - \tau)(X_i - \bar{X})^\T. 
\end{align*}
Here, $S^2_{z \mid X}$ is the variance of the linear projection of potential outcomes on $X$, and $S^2_{\tau \mid X}$ is analogously defined for treatment effects. The asymptotic distribution of $\hat{\tau}$ under rerandomization is \citep{li2018asymptotic}:
\begin{align}
    V^{-1/2} N^{1/2}\left( \hat{\tau} - \tau \right) \mid M \leq a \sim (1 - R^2)^{1/2} \epsilon_0 + R L_{K,a}, \label{eqn:rerandDistribution}
\end{align}
where $\epsilon_0 \sim \mathcal{N}(0,1)$, $R^2$ is the squared multiple correlation between $X$ and potential outcomes: 
\begin{align}
    R^2 = \frac{p_1^{-1} S^2_{1\mid X} + p_0^{-1} S_{0\mid X}^2 - S^2_{\tau\mid X}}{p_1^{-1} S^2_1 + p_0^{-1} S_0^2 - S^2_{\tau}}, \label{eqn:r2}
\end{align}
and $L_{K,a}\sim \chi_{K,a} S \beta_K^{1/2}$ with
\begin{equation}
	\chi^2_{K,a} \sim \chi^2_K  \mid  \chi^2_K \leq a, \quad 
	S \sim -1 + 2 \cdot \text{Bern}(1/2), \quad 
	\beta_K \sim \text{Beta}\{1/2, (K-1)/2\} . \label{eqn:representation}
\end{equation}
The rerandomization distribution (\ref{eqn:rerandDistribution}) and complete randomization distribution (\ref{eqn:compRandDistribution}) are identical only if $R^2 = 0$ or $a=\infty$. Otherwise, (\ref{eqn:rerandDistribution}) will be a non-Normal distribution that has less variance than a standard Normal distribution \citep{li2018asymptotic}. Both distributions involve $V$, which depends on the proportions $p_1$ and $p_0$ and variances $S_1^2$, $S_0^2$, and $S^2_\tau$. The proportions are fixed; meanwhile, $S_1^2$ and $S_0^2$ can be consistently estimated but $S^2_{\tau}$ cannot without additional assumptions. In practice, we can estimate $V$ conservatively. \citet{Neyman1923} proposed $\hat{V}_{\neyman} = p_1^{-1} s^2_1 + p_0^{-1} s_0^2$, where $s_1^2$ and $s_0^2$ are sample versions of $S^2_1$ and $S^2_0$. The estimator $\hat{V}_{\neyman}$ implicitly estimates treatment effect heterogeneity as $\hat{S}^2_{\tau} = 0$ and is thus conservative, in the sense that $E(\hat{V}_{\neyman} - V) = S^2_{\tau} \ge 0$. \citet{ding2019decomposing} noted that $S^2_{\tau} \geq S^2_{\tau\mid X}$, which can be consistently estimated by $s^2_{\tau\mid X} = (s_{1,X} - s_{0, X})(S_X^2)^{-1}(s_{X,1} - s_{X,0})$ \citep{li2018asymptotic}, and thus proposed an improved variance estimator $\hat{V}_{\ding} = \hat{V}_{\neyman} - s^2_{\tau\mid X}$. Here, $s_{1,X} = s_{X,1}^\T $ and $s_{0,X} = s_{X,0}^\T$ are sample analogs of $S_{1,X}$ and $s_{0,X}$, respectively. We will consider both estimators and use the generic notation $\hat{V}$. 
As demonstrated in \citet{li2017general} and \citet{ding2019decomposing}, $\hat{V}$ has a probability limit $\tilde{V}$ no less than the true variance $V$, i.e., $\hat{V} = \tilde{V} + o_{\pr}(1)$ with $\tilde{V} \ge V$.
The probability limits of $\hat{V}_{\neyman}$ and $\hat{V}_{\ding}$ are, respectively,
\begin{align}\label{eqn:V_tilde}
    \tilde{V}_{\neyman} = p_1^{-1} S_1^2 + p_0^{-1} S_0^2, 
    \quad
    \tilde{V}_{\ding} = p_1^{-1} S_1^2 + p_0^{-1} S_0^2 - S_{\tau \mid X}^2.
\end{align}
Thus, $\hat{V}_{\neyman}$ is consistent when $S^2_{\tau}=0$, 
and $\hat{V}_{\ding}$ is consistent more broadly when the individual effects can be linearly explained by the covariates.

The rerandomization distribution (\ref{eqn:rerandDistribution}) suggests the following $(1-\alpha)$-level confidence interval:
\begin{align}
	&\hat{\tau} \pm  \nu_{1-\alpha/2}(\hat{R}^2) \hat{V}^{1/2} N^{-1/2}, \label{eqn:rerandCI}
\end{align}
where $\nu_{1-\alpha/2}(\rho^2)$ denotes the ($1-\alpha/2$)-quantile of the distribution $(1 - \rho^2)^{1/2} \epsilon_0 + \rho L_{K,a}$. 
The representation (\ref{eqn:representation}) makes it simple to approximate this quantile via Monte Carlo after specifying $\hat{R}^2$. In \eqref{eqn:rerandCI}, $\hat{R}^2 = (p_1^{-1} s^2_{1\mid X} + p_0^{-1} s_{0\mid X}^2 - s^2_{\tau\mid X})/\hat{V}$, where $s^2_{1\mid X}$, $s^2_{0\mid X}$, and $s^2_{\tau \mid X}$ are sample analogues of $S^2_{1\mid X}$, $S^2_{0\mid X}$, and $S^2_{\tau\mid X}$, as defined in \cite{li2018asymptotic}. The estimator $\hat{R}^2$ has probability limit $\tilde{R}^2 = VR^2/\tilde{V} \le R^2$;
thus, $\hat{R}^2$ is conservative to the extent that $\hat{V}$ is conservative. Because $\nu_{1-\alpha/2}(R^2) \leq z_{1-\alpha/2}$ for all $\alpha \in (0, 1)$, the interval (\ref{eqn:rerandCI}) is narrower than the analogous interval for a completely randomized experiment if the covariates have any linear association with the outcomes (\citealt[Theorem 2]{li2018asymptotic}).

The interval (\ref{eqn:rerandCI}) implies the following test for the null $H_0: \tau = 0$ against the alternative $H_A: \tau > 0$:
\begin{align}
\begin{cases}
\text{Reject $H_0: \tau = 0$} &\mbox{ if } \hat{\tau} > \nu_{1-\osl}(\hat{R}^2) \hat{V}^{1/2} N^{-1/2}, \\
\text{Fail to reject $H_0: \tau = 0$} &\mbox{ otherwise.}
\end{cases}
\label{eqn:rerandTest}    
\end{align}
Similar to most power analyses \citep{lachin1981introduction,lerman1996study,wittes2002sample}, we focus on the one-sided test here because it simplifies sample size calculations. For any $\alpha$-level two-sided test, we can always bound its power and sample size requirement by that of a $\alpha/2$-level one-sided test. 
Furthermore, 
throughout the paper, we assume the significance level $\alpha$ of the test  in \eqref{eqn:rerandTest} is below or equal to $0.5$, which is the case for most if not all applications.
The following theorem establishes the power of the above test.

\begin{theorem}
\label{thm:powerRerand}
	Under rerandomization, the power of the test (\ref{eqn:rerandTest}) is asymptotically:
	\begin{align*}
	    \overline{\mathcal{V}}_{R^2}\left( \frac{\nu_{1-\osl} (\tilde{R}^2) \tilde{V}^{1/2} - \tau N^{1/2}}{V^{1/2} } \right),
	\end{align*}
	where 
	$\overline{\mathcal{V}}_{R^2}(\cdot)$ denotes the survival function of $(1 - R^2)^{1/2} \epsilon_0 + R L_{K,a}$, 
	$V$ denotes the variance in (\ref{eqn:compRandDistribution}), $\tilde{V}$ denotes the probability limit of its corresponding estimator, $R^2$ denotes the squared multiple correlation in (\ref{eqn:r2}), and $\tilde{R}^2 = VR^2/\tilde{V}$ denotes the probability limit of $\hat{R}^2$ in (\ref{eqn:rerandCI}). 
\end{theorem}
Theorem \ref{thm:powerRerand} is quite similar to classical power calculations (e.g., \citealt{lachin1981introduction}), with two important differences. First, power for completely randomized experiments rely on the standard Normal quantile $z_{1-\osl}$ and survival function $\overline{\Phi}(\cdot)$, whereas Theorem \ref{thm:powerRerand} relies on the non-Normal quantile $\nu_{1-\osl}(\tilde{R}^2)$ and survival function $\overline{\mathcal{V}}_{R^2}(\cdot)$.
Second, Theorem \ref{thm:powerRerand} involves the ratio $\tilde{V}/V$, which depends on the treatment effect heterogeneity $S^2_\tau$. 
Such a term does not appear in previous power results \citep{lachin1981introduction,cohen1992power,lerman1996study,wittes2002sample}, 
which typically assume a super-population framework that does not involve $S^2_{\tau}$.
However, under a finite-population framework, $S^2_{\tau}$ appears in the true variance of $\hat{\tau}$ and thereby affects power calculations. Thus, Theorem \ref{thm:powerRerand} allows for less conservative variance estimators, such as $\hat{V}_{\ding}$, and demonstrates how treatment effect heterogeneity affects testing power.

We can show that, for fixed $S^2_1$, $S^2_0$, and $\tau$, power in Theorem \ref{thm:powerRerand} is increasing in $S^2_{\tau}$ as long as $\tau \geq \nu_{1-\osl}(\tilde{R}^2) \tilde{V}^{1/2} N^{-1/2}$; otherwise, it is decreasing in $S^2_{\tau}$. Thus, treatment effect heterogeneity has a beneficial effect on power for large effect sizes but an adverse effect for small effect sizes. As a result, standard power calculations that assume $S^2_\tau = 0$ may underestimate or overestimate power, depending on the effect size. We discuss and illustrate this dynamic further in the supplementary material.

\subsection{Sample Size Calculations} \label{ss:rerandSampSize}

Theorem \ref{thm:powerRerand} establishes the testing power of a rerandomized experiment.
The following theorem establishes the relationship between the sample size $N$ and a prespecified degree of power $\gamma$ when we use the test (\ref{eqn:rerandTest}) to conduct inference for a rerandomized experiment.
\begin{theorem}
\label{thm:sampSizeRerand}
	Let $\gamma \geq \osl$ denote a prespecified probability that we correctly reject $H_0: \tau = 0$ using the test (\ref{eqn:rerandTest}) under rerandomization. Then, the relationship between the sample size $N$ and power $\gamma$ is:
	\begin{align}
	   N &= \left\{\frac{\nu_{1-\osl}(\tilde{R}^2) \tilde{V}^{1/2} - \nu_{1 - \gamma}(R^2) V^{1/2} }{\tau} \right\}^2, \label{eqn:rerandSampSize}
	\end{align}
	where $V$ denotes the variance in (\ref{eqn:compRandDistribution}), $\tilde{V}$ denotes the probability limit of its corresponding estimator, $R^2$ denotes the squared multiple correlation in (\ref{eqn:r2}), and $\tilde{R}^2$ denotes the probability limit of $\hat{R}^2$ in (\ref{eqn:rerandCI}). 
\end{theorem}

\noindent
The sample size in Theorem \ref{thm:sampSizeRerand} is (a) increasing in the power $\gamma$, (b) decreasing in the average treatment effect $\tau$, (c) increasing in the potential outcome variances $S^2_1$ and $S^2_0$ if $\gamma \geq 0.5$, and (d) decreasing in the treatment effect heterogeneity $S^2_{\tau}$ if $\gamma \geq 0.5$. These results also hold under complete randomization, which is a special case of rerandomization when $a = \infty$ or $R^2 = 0$. The first three observations are well-known, but, to our knowledge, the fourth has remained largely unacknowledged. However, it is intuitive: If $S^2_\tau$ is larger, then the variance of the mean-difference estimator is smaller, as shown in (\ref{eqn:compRandDistribution}). This additional precision is propagated into the sample size necessary to achieve power $\gamma$. This also demonstrates that assuming $S^2_{\tau} = 0$ is conservative, not only for confidence intervals, but also for sample size calculations. If one has a priori knowledge about $S^2_\tau$, then as long as the desired power is greater than 0.5, one could use Theorem \ref{thm:sampSizeRerand} to argue that the required sample size is less than what is recommended by classic power calculations, which assume $S^2_\tau = 0$. However, this argument should be proceeded with caution, because typically there is no knowledge about the degree of treatment effect heterogeneity that will occur. Furthermore, if increased heterogeneity in turn increases potential outcome variation, this could have an adverse effect on sample size. We discuss this nuance further in the supplementary material, where we study via simulation how sample size is affected by potential outcome variation and treatment effect heterogeneity.

\begin{remark}
Technically, the sample size $N$ is also on the right-hand side of (\ref{eqn:rerandSampSize}), because $N$ is involved in the definitions of $S^2_1$, $S^2_0$, and $S^2_{\tau}$. This is a by-product of adopting a finite-population framework. Nonetheless, we write $N$ in terms of $S^2_1$, $S^2_0$, and $S^2_{\tau}$, because it is common to specify these quantities when making sample size calculations (e.g., \citealt{lachin1981introduction,wittes2002sample}). With a slight abuse of notation, we can view $S^2_1$, $S^2_0$, and $S^2_{\tau}$ in Theorem \ref{thm:sampSizeRerand} as limits of potential outcome variances and treatment effect heterogeneity, thereby allowing for sample size calculators under a finite-population framework. 
\end{remark}

\section{Comparing Rerandomization to Complete Randomization} \label{s:comparison}

\subsection{Dispersive Ordering of Normal and Non-Normal Distributions}

A natural question is how power and sample size calculations for rerandomization compare to classical results for complete randomization. \cite{li2018asymptotic} established that two-sided confidence intervals under rerandomization are narrower than those under complete randomization. 
To do this, they compared the lengths of symmetric quantile ranges $[\nu_{\alpha}(R^2), \nu_{1-\alpha}(R^2)]$ and $[z_{\alpha}, z_{1-\alpha}]$ for $\alpha \in (0,0.5]$. However, such a comparison is not sufficient for comparing power between rerandomization and complete randomization, because the power comparison involves lengths of asymmetric quantiles ranges, and in particular, gaps between two quantiles on the same side of the origin. 

To make this comparison, we establish a dispersive ordering of Normal and non-Normal distributions involved in complete randomization and rerandomization. For two random variables $X$ and $Y$ with quantile functions $F^{-1}$ and $G^{-1}$, 
$X$ is said to be less dispersed than $Y$ if 
$F^{-1}(\beta) - F^{-1}(\alpha) \le G^{-1}(\beta) - G^{-1}(\alpha)$ for any $0<\alpha<\beta<1$ \citep{shaked1982}. The following theorem establishes that the rerandomization distribution (\ref{eqn:rerandDistribution}) is less dispersed than a standard Normal distribution.

\begin{theorem}\label{thm:disp_std_rem}
Let $\varepsilon_0 \sim \mathcal{N}(0,1)$ and $L_{K,a}$ be defined in (\ref{eqn:representation}), such that $\varepsilon_0$ and $L_{K,a}$ are independent.
Then, for any $a \in [0, \infty]$, integer $K\ge 1$ and $\rho \in [0,1]$, 
$(1-\rho^2)^{1/2} \varepsilon_0 + \rho L_{K,a}$ is less dispersed than $\varepsilon_0$.  
\end{theorem}
Theorem \ref{thm:disp_std_rem} generalizes \citet[][Theorem 2]{li2018asymptotic}, which showed that the symmetric quantile range of $(1-\rho^2)^{1/2} \varepsilon_0 + \rho L_{K,a}$ is less than or equal to that of the standard Normal distribution. Theorem \ref{thm:disp_std_rem} further shows that the gap between the $\alpha$- and $\beta$-quantiles of $(1-\rho^2)^{1/2} \varepsilon_0 + \rho L_{K,a}$, such as any non-symmetric quantile range, is less than or equal to that of the standard Normal distribution, for any $\alpha, \beta \in (0, 1)$. Theorem \ref{thm:disp_std_rem} is crucial for establishing the later theorems on power and sample size comparisons.

\subsection{Power and Sample Size Comparisons}\label{sec:power_samplesize}

The following theorem quantifies when power is greater under rerandomization.
\begin{theorem}\label{thm:power_compare}
If $\tilde{V} = V$, then for any $\tau \ge 0$, 
    \begin{align}\label{eq:power_bound}
        \overline{\mathcal{V}}_{R^2}\left( \frac{\nu_{1-\osl} (\tilde{R}^2) \tilde{V}^{1/2} - \tau N^{1/2}}{V^{1/2}} \right)
        \ge 
        \overline{\Phi}\left( \frac{z_{1-\osl} \tilde{V}^{1/2} - \tau N^{1/2}}{V^{1/2} } \right).
    \end{align}
    Meanwhile, if $\tilde{V} > V$, \eqref{eq:power_bound} still holds when $\tau \ge \nu_{1-\osl} (\tilde{R}^2) \tilde{V}^{1/2} N^{-1/2}$; 
    otherwise, \eqref{eq:power_bound} may be violated.
\end{theorem}
\noindent
From Theorem \ref{thm:power_compare}, if inference is not conservative, such that $\tilde{V} = V$, power is greater under rerandomization than under complete randomization. However, when inference is conservative, such that $\tilde{V} > V$, rerandomization may exhibit less power than complete randomization.
This result is surprising: Confidence intervals are always narrower under rerandomization, so we would suspect power to always be greater. 
However, for power, we must consider not only confidence intervals' precision but also their conservativeness.
Equation (\ref{eqn:rerandDistribution}) shows that the true asymptotic distribution of $N^{1/2}(\hat{\tau}-\tau)$ under rerandomization is $V^{1/2}\{ (1-R^2)^{1/2} \epsilon_0 + RL_{k,a} \}$; when conducting inference, this is asymptotically estimated as $\tilde{V}^{1/2}\{ (1-\tilde{R}^2)^{1/2} \epsilon_0 + \tilde{R}L_{k,a} \}$. By recognizing that $\tilde{V}\tilde{R}^2 = VR^2$, one can show that this estimated distribution is equivalent to the convolution of the true distribution of $N^{1/2}(\hat{\tau}-\tau)$ and an independent Normal distribution $\mathcal{N}(0, \tilde{V} - V)$. Because the true distribution of $N^{1/2}(\hat{\tau} - \tau)$ is more concentrated around zero under rerandomization than complete randomization but the distribution $\mathcal{N}(0, \tilde{V} - V)$ is the same for both designs, inference under rerandomization can be relatively more conservative. As a result, the test under rerandomization can have a smaller type-I error; see the supplementary material for details. This additional conservativeness has an adverse effect on power, but the additional precision from rerandomization has a beneficial effect. Theorem \ref{thm:power_compare} establishes that the beneficial effect outweighs the adverse effect as long as $\tau$ is not too small. We illustrate this trade-off in the supplementary material.

Intuitively, when rerandomization increases power, it requires a smaller sample size to achieve a certain degree of power. Let $N_{\textup{rr}}$ denote the sample size necessary to achieve power $\gamma$ under rerandomization, as provided by Theorem \ref{thm:sampSizeRerand}. Meanwhile, let $N_{\textup{cr}}$ denote the sample size necessary under complete randomization, which is the same as Theorem \ref{thm:sampSizeRerand}, but with the non-Normal quantiles replaced with Normal quantiles. The following theorem establishes sufficient conditions for $N_{\textup{rr}} \leq N_{\textup{cr}}$.
\begin{theorem}
\label{thm:sampleSizeRatio}
	Let $\gamma \geq \osl$ denote a desired level of power, and let $N_{\textup{rr}}$ and $N_{\textup{cr}}$ denote the sample sizes required to achieve power $\gamma$ under rerandomization and complete randomization, respectively. We have two separate results. First, if $\tilde{V} = V$, then $N_{\textup{rr}}/N_{\textup{cr}} \leq 1$. Second, if $\gamma\ge 0.5$, then $N_{\textup{rr}}/N_{\textup{cr}}$ is (a) less than or equal to 1, (b) increasing in the number of covariates $K$ and the rerandomization threshold $a$, and (c) decreasing in $R^2$. Otherwise, if $\tilde{V} > V$ and $\gamma < 0.5$, then $N_{\textup{rr}}/N_{\textup{cr}}$ may be greater than 1.
\end{theorem}

Theorem \ref{thm:sampleSizeRatio} establishes that rerandomization requires a smaller sample size to achieve the same amount of power when either inference is not conservative or the desired power is greater than 0.5. The condition $\tau \geq \nu_{1-\osl}(\tilde{R}^2) \tilde{V}^{1/2}N^{-1/2}$ in Theorem \ref{thm:power_compare} implies the condition $\gamma \geq 0.5$ in Theorem \ref{thm:sampleSizeRatio}, and thus these conditions are analogous. In most standard power calculations, the desired power is greater than 0.5.

The smaller the ratio $N_{\textup{rr}}/N_{\textup{cr}}$, the larger the sample size benefits of rerandomization over complete randomization. The sample size reduction depends on $R^2$, $K$, $a$, and $S^2_{\tau}$. In the supplementary material, we conduct a simulation study to assess how $N_{\textup{rr}}/N_{\textup{cr}}$ changes for different $K, R^2, a$, and $S^2_{\tau}$, thereby allowing practitioners to understand the sample size benefits of rerandomization. When there is no treatment effect heterogeneity, the median of $N_{\textup{rr}}/N_{\textup{cr}}$ is 0.75 for $K \in [1, 100]$ and $R^2 \in [0, 0.9]$ using a rerandomization threshold $a$ such that $\pr(M \leq a) = 0.001$; if further $R^2 \geq 0.3$ and $K \leq 50$, the median is 0.58. Meanwhile, when there is treatment effect heterogeneity, these sample size benefits are dampened, because inference is conservative. When $S_{\tau}$ is half the size of $S_1$ and $S_0$, $N_{\textup{rr}}/N_{\textup{cr}}$ is on average 2.4\% greater than when there is no heterogeneity; and when $S_{\tau}$ is 50\% greater than $S_1$ and $S_0$, $N_{\textup{rr}}/N_{\textup{cr}}$ is on average 23.7\% greater. More details, such as how to implement power and sample size calculations with our \texttt{R} package \texttt{rerandPower}, are discussed in the supplementary material.

\section{Discussion and Conclusion} \label{s:conclusion}

Our results focus on rerandomized experiments with two groups. Rerandomization theory has been extended to more than two groups \citep{branson2016improving,li2020rerandomizationFactorial}, and thus we posit that similar results hold for multi-arm experiments. However, multiple causal estimands arise in this setting, making power analyses more complex. Power will depend on the potential outcome variance in each group, as well as the effect heterogeneity and rerandomization criterion for each estimand. Thus, power analyses will be notationally complex, but they can rely on the same conceptual framework developed here.

One could compare rerandomization to designs beyond complete randomization, such as block randomization, 
where blocks are constructed using, e.g., matching \citep{greevy2004optimal,bai2022optimality}. Power would then depend on the association between blocking variables and outcomes; however, blocking can be less efficient than complete randomization if the blocks are poorly chosen \citep{pashley2022block}. 
Meanwhile, rerandomization is always at least as efficient as complete randomization and can be combined with blocking to improve block randomization \citep{wang2021rerandomization}. 
Comparing block randomization and block rerandomization is analogous to our comparison here, and we leave it for future study.

Our results also focus on the Mahalanobis distance on covariate means, but other rerandomization criteria could be used. For example, \cite{morgan2015rerandomization} proposed different Mahalanobis distance criteria for tiers of covariates that vary in importance. Again power analyses will be notationally complex in order to incorporate the criterion for each tier. Other examples include criteria modified by ridge penalties \citep{branson211ridge} or principal component analysis \citep{zhang2021pca}, which have been shown to increase precision in high-dimensional settings. We suspect that testing power may increase as well.


\newpage

\begin{center}
	\LARGE{Supplementary Material}
\end{center}

\appendix

\renewcommand{\theequation}{S.\arabic{equation}}


Appendix \ref{s:proofPowerRerand} provides the proof of Theorem \ref{thm:powerRerand}, which characterizes the asymptotic power under rerandomization. 

Appendix \ref{s:proofSampSizeRerand} provides the proof of Theorem \ref{thm:sampSizeRerand}, which characterizes the sample size necessary to achieve a given power under rerandomization. 

Appendix \ref{s:dispersiveOrdering} provides the proof of Theorem \ref{thm:disp_std_rem}, which establishes the dispersive ordering of the Normal and non-Normal distributions involved in the asymptotic approximations for complete randomization and rerandomization. 

Appendix \ref{s:powerCompare} provides the proof of Theorem \ref{thm:power_compare}, which determines when asymptotic power is greater under rerandomization than under complete randomization. 

Appendix \ref{s:proofSampSizeRatio} provides the proof of Theorem \ref{thm:sampleSizeRatio}, which characterizes the ratio of the rerandomization sample size and complete randomization sample size necessary to achieve a given power.

Appendix \ref{sec:compare_actual_typeI} compares the type-I error rates of rerandomization and complete randomization under the null hypothesis of $\tau = 0$. 

Appendix \ref{sec:numerical_power_rerand_less} uses numerical examples to illustrate that rerandomization can be less powerful than complete randomization.

Appendix \ref{s:simulations} presents a simulation study exploring how 
$N_{\textup{rr}} / N_{\textup{cr}}$
changes for various experimental settings, thereby allowing practitioners to understand the sample size gains of rerandomization compared to complete randomization. 

Appendix \ref{s:numericalExamples} provides additional numerical examples that illustrate how treatment effect heterogeneity affects power and sample size under complete randomization and rerandomization.

Appendix \ref{s:package} provides example code to implement power and sample size calculations for completely randomized and rerandomized experiments with our \texttt{R} package \texttt{rerandPower}.

\section{Proof of Theorem \ref{thm:powerRerand}} \label{s:proofPowerRerand}

Below we first give a rigorous statement of Theorem \ref{thm:powerRerand} in the main paper. 
We will conduct finite population asymptotic analyses; see, e.g., \citet{li2017general} for a review. 
Specifically, 
we embed the $N$ units into a sequence of finite populations and impose the following regularity conditions analogous to \citet[][Condition 1]{li2018asymptotic} along this sequence of finite populations. 
\begin{condition}\label{cond:fp}
	As $N\rightarrow \infty$,
	\begin{itemize}
		\item[(i)] the proportions of treated and control  units, $p_1$ and $p_0$, have positive limits; 
		\item[(ii)] the finite population variances and covariances, $S_1^2, S_0^2, S^2_{\tau}, S_{1, X}, S_{0,X}$ and $S_X^2$ have finite limiting values, the limit of $S_X^2$ is nonsingular, and the limit of $V = p_1^{-1} S_1^2 + p_0^{-1} S_0^2 - S^2_{\tau}$ is positive. 
		\item[(iii)] $\max_{1\le i \le N} \{Y_i(z) - \bar{Y}(z)\}^2/N \rightarrow 0$ for $z=0,1$, 
		and $\max_{1\le i \le N} \|X_i^2 - \bar{X}\|^2/N \rightarrow 0$.
	\end{itemize}
\end{condition}

In Condition \ref{cond:fp}, (i) is a natural requirement, and (ii) assumes stable finite population variances and covariances along the sequence of finite populations. In addition, (ii) intuitively assumes that the covariates are not colinear, and that the variance of $\hat{\tau}$ under complete randomization is not zero. 
Meanwhile, (iii) assumes that  the potential outcomes and covariates are not too heavy-tailed, and it will hold with probability $1$ if the units are i.i.d.\ samples from some superpopulation with $2+\eta$ moments, for any $\eta > 0$. 

The theorem below is a rigorous version of Theorem \ref{thm:powerRerand} in the main paper. 
In this asymptotic power analysis, we consider the finite population asymptotics with Condition \ref{cond:fp} and the local alternative where the true average treatment effect is on the scale of $N^{-1/2}$, analogous to usual power analysis. 

\begin{theorem*}
	Consider a sequence of finite populations with increasing sizes. Assume that Condition \ref{cond:fp} holds and the true average treatment effects satisfy $\tau = c N^{-1/2}$ for all $N$ and some finite constant $c$. 
	Under rerandomization, the power of the test \eqref{eqn:rerandTest} satisfies 
	\begin{align*}
		\pr\left( \hat{\tau} > \nu_{1-\osl}(\hat{R}^2) \hat{V}^{1/2} N^{-1/2} \right) 
		 = \overline{\mathcal{V}}_{R^2}\left( \frac{\nu_{1-\osl}(\tilde{R}^2) \tilde{V}^{1/2} - \tau N^{1/2}}{V^{1/2}} \right) + o(1), 
	\end{align*}
	where $\overline{\mathcal{V}}_{R^2}(\cdot)$ denotes the survival function of $(1 - R^2)^{1/2} \epsilon_0 + R L_{K,a}$, 
	$V$ denotes the variance in (\ref{eqn:compRandDistribution}), $\tilde{V}$ denotes the probability limit of its corresponding estimator, $R^2$ denotes the squared multiple correlation in (\ref{eqn:r2}), and $\tilde{R}^2 = VR^2/\tilde{V}$ denotes the probability limit of $\hat{R}^2$ in (\ref{eqn:rerandCI}). 
\end{theorem*}

Now we are going to prove Theorem \ref{thm:powerRerand}. 
Under the test (\ref{eqn:rerandTest}), we reject the null hypothesis if $\hat{\tau} > \nu_{1-\osl}(\hat{R}^2) \hat{V}^{1/2} N^{-1/2}$. 
Thus, the power of the test is 
\begin{align*}
	\pr\left( \hat{\tau} > \nu_{1-\osl}(\hat{R}^2) \hat{V}^{1/2} N^{-1/2} \right)
	& = 
	\pr\left( N^{1/2}(\hat{\tau} - \tau) - \nu_{1-\osl}(\hat{R}^2) \hat{V}^{1/2} > - N^{1/2}\tau \right)\\
	& = \pr\left( N^{1/2}(\hat{\tau} - \tau) - \nu_{1-\osl}(\hat{R}^2) \hat{V}^{1/2} >  c \right),
\end{align*}
where the last equality holds due to the fact that $\tau = N^{-1/2} c$. 
From \citet{li2018asymptotic}, under rerandomization, 
$
\sqrt{N}(\hat{\tau} - \tau) \ \dot \sim \ V^{1/2} \{(1-R^2)^{1/2} \varepsilon_0 + R L_{K,a}\},
$
where $\dot\sim$ denotes that the two sequences of distributions or random variables converging weakly to the same distribution.  
Besides, $\tilde{V}-V = o_{\pr}(1)$ and $R^2 - \tilde{R}^2 = o_{\pr}(1)$. 
By Slutsky's theorem, we have $N^{1/2}(\hat{\tau} - \tau) - \nu_{1-\osl}(\hat{R}^2) \hat{V}^{1/2}\ \dot\sim \ V^{1/2} \{(1-R^2)^{1/2} \varepsilon_0 + R L_{K,a}\} - \nu_{1-\osl}(\tilde{R}^2) \tilde{V}^{1/2}$. 
Consequently, we have, as $N\rightarrow \infty$, 
\begin{align*}
	\pr\left( \hat{\tau} > \nu_{1-\osl}(\hat{R}^2) \hat{V}^{1/2} N^{-1/2} \right) 
	& = \pr\left( V^{1/2} \{(1-R^2)^{1/2} \varepsilon_0 + R L_{K,a}\} - \nu_{1-\osl}(\tilde{R}^2) \tilde{V}^{1/2} >  c \right) + o(1)\\
	& = 
	\overline{\mathcal{V}}_{R^2}\left( \frac{\nu_{1-\osl}(\tilde{R}^2) \tilde{V}^{1/2} - c}{V^{1/2}} \right) + o(1)\\
	& = \overline{\mathcal{V}}_{R^2}\left( \frac{\nu_{1-\osl}(\tilde{R}^2) \tilde{V}^{1/2} - \tau N^{1/2}}{V^{1/2}} \right) + o(1).
\end{align*}
Therefore, Theorem \ref{thm:powerRerand} holds.

\section{Proof of Theorem \ref{thm:sampSizeRerand}} \label{s:proofSampSizeRerand}

Let $\gamma$ denote a prespecified degree of power desired for a rerandomized experiment, where $\gamma$ is the probability we reject the null hypothesis using the test (\ref{eqn:rerandTest}). Then, by Theorem \ref{thm:powerRerand}, we have:
\begin{align*}
	\gamma = \overline{\mathcal{V}}_{R^2}\left( \frac{\nu_{1-\osl}(\tilde{R}^2) \tilde{V}^{1/2} - \tau N^{1/2}}{V^{1/2}} \right) = 1 - \mathcal{V}_{R^2}\left( \frac{\nu_{1-\osl}(\tilde{R}^2) \tilde{V}^{1/2} - \tau N^{1/2}}{V^{1/2}} \right), 
\end{align*}
where $\mathcal{V}_{R^2}(\cdot)$ denotes the distribution function of the distribution $(1 - R^2)^{1/2} \epsilon_0 + R L_{K,a}$ defined in (\ref{eqn:rerandDistribution}). Then, solving for $N$, we have:
\begin{align*}
    && \frac{\nu_{1-\osl}(\tilde{R}^2) \tilde{V}^{1/2} - \tau N^{1/2}}{V^{1/2}}  = \nu_{1-\gamma}(R^2) \\
  &  \Longrightarrow &  \nu_{1-\osl}(\tilde{R}^2) \tilde{V}^{1/2} - \tau N^{1/2}  = \nu_{1-\gamma}(R^2)V^{1/2} \\
   & \Longrightarrow &  N  = \left( \frac{\nu_{1-\osl}(\tilde{R}^2) \tilde{V}^{1/2} - \nu_{1-\gamma}(R^2) V^{1/2}}{\tau} \right)^2 .
\end{align*}

\section{Proof of Theorem \ref{thm:disp_std_rem}} \label{s:dispersiveOrdering}

To prove Theorem \ref{thm:disp_std_rem}, we need the following five lemmas.

\begin{lemma}\label{lemma:Ldensity}
For any integer $K\ge 1$ and threshold $a\in (0, \infty)$, 
the probability density function of $L_{K,a}$ is 
\begin{align*}
    g_{K,a}(x) = \phi(x) \frac{F_{K-1}(a-x^2)}{F_K(a)}, 
\end{align*}
where $\phi(\cdot)$ is the probability density of $\mathcal{N}(0,1)$ and $F_K(\cdot)$ is the distribution function of $\chi^2_K$, with $F_0(x) = \I(x\ge 0)$ being the distribution function of a point mass at $0$. 
\end{lemma}

\begin{proof}[of Lemma \ref{lemma:Ldensity}]
Lemma \ref{lemma:Ldensity} follows from \citet[][Proof of Proposition 2]{li2018asymptotic}. For completeness, we give a proof below. 
Let $\bs{D} = (D_1, \ldots, D_K)^\T \sim \mathcal{N}(\bs{0}, \bs{I}_K)$. 
For 
any $x\in \mathbb{R}$, 
we have
\begin{align*}
    \pr(L_{K,a} \le x) 
    & = 
    \pr(D_1 \le x \mid \bs{D}^\T \bs{D} \le a) \\
  &  = 
    \frac{\pr(D_1 \le x, \bs{D}^\T \bs{D} \le a)}{\pr(\bs{D}^\T \bs{D} \le a)}
    \\
    & = 
    \frac{1}{F_K(a)} \int_{-\infty}^{\infty} \pr\left(t \le x, t^2 + \sum_{j=2}^K D_j^2 \le a \right) \phi(t) \deri t  \\
&    = 
    \frac{1}{F_{K}(a)} \int_{-\infty}^{x} F_{K-1}(a-t^2) \phi(t) \deri t\\
    & \equiv \int_{-\infty}^{x} g_{K,a}(t) \deri t, 
\end{align*}
where $g_{K,a}(t) \equiv F_{K-1}(a-t^2) \phi(t) / F_{K}(a)$. 
Therefore, $g_{K,a}(\cdot)$ must be the probability density function of $L_{K,a}$, i.e., Lemma \ref{lemma:Ldensity} holds. 
\end{proof}

\begin{lemma}\label{lemma:diff_logdensity_L_epsilon}
For any $a\in (0,\infty)$, $K\ge 1$ and $c\in \mathbb{R}$, define 
\begin{align*}
    h_{K,a,c} (x) & = \log g_{K,a}(x) - \log \phi(x+c), \qquad (-\sqrt{a} < x < \sqrt{a}). 
\end{align*}
Then $\deri^2 h_{K,a,c} (x) / \deri x^2 \le 0$ for $x\in (-\sqrt{a}, \sqrt{a})$.
\end{lemma}

\begin{proof}[of Lemma \ref{lemma:diff_logdensity_L_epsilon}]
From Lemma \ref{lemma:Ldensity}, 
\begin{align*}
    h_{K,a,c} (x) & = \log g_{K,a}(x) - \log \phi(x+c)
    = 
    \log \phi(x) + \log F_{K-1}(a-x^2) - \log F_K(a) - \log \phi(x+c)\\
    & = \log F_{K-1}(a-x^2) + cx + c^2/2 - \log F_K(a). 
\end{align*}
When $K=1$, $h_{1,a,c} (x)$ reduces to $cx+c^2/2 - \log F_1(a)$, which is a linear function of $x$. Consequently, $\deri^2 h_{1,a,c}(x)/\deri x^2 = 0$ for all $x\in (-\sqrt{a}, \sqrt{a})$, i.e., Lemma \ref{lemma:diff_logdensity_L_epsilon} holds for $K=1$. 
Below we consider only the case with $K > 1$.

Let $f_K(x)$ be the density of $\chi^2_K$, and $\dot{f}_K(x) = \deri f_K(x)/ \deri x$ be its derivative over $x$. 
We have
\begin{align}\label{eq:deriv_chi_density}
    \dot{f}_K(x) = f_K(x) \cdot \left( \frac{K/2-1}{x} - \frac{1}{2}\right)
    = f_K(x) \cdot \frac{K-2-x}{2x}. 
\end{align}
Consequently, for $K > 1$, the second derivative of $h_{K,a,c}$ reduces to
\begin{align*}
    \frac{\deri^2}{\deri x^2} h_{K,a,c} (x)
    & = \frac{\deri^2}{\deri x^2} \log F_{K-1}(a-x^2)
    \\
    & = \frac{\deri}{\deri x} 
    \left\{ 
    \frac{f_{K-1}(a-x^2) \cdot (-2x)}{F_{K-1}(a-x^2)} 
    \right\}\\
    & = 
    \frac{\dot{f}_{K-1}(a-x^2) \cdot (-2x)^2 \cdot F_{K-1}(a-x^2) - \{f_{K-1}(a-x^2) \cdot (-2x)\}^2}{\{F_{K-1}(a-x^2)\}^2}
    \\
    & = 
    \frac{4x^2\cdot f_{K-1}(a-x^2)}{\{F_{K-1}(a-x^2)\}^2}
    \left\{
    \frac{K-1-2-(a-x^2)}{2(a-x^2)} \cdot F_{K-1}(a-x^2) - f_{K-1}(a-x^2)
    \right\}
    \\
    & \equiv
    \frac{4x^2\cdot f_{K-1}(a-x^2)}{\{F_{K-1}(a-x^2)\}^2} \cdot \Delta_{K-1}(a-x^2), 
\end{align*}
where 
\begin{align*}
    \Delta_K(x) = \frac{K-2-x}{2x} \cdot F_{K}(x) - f_{K}(x). 
\end{align*}
Thus, to prove Lemma \ref{lemma:diff_logdensity_L_epsilon}, it suffices to prove $\Delta_K(x) \le 0$ for all $K > 0$ and $x\in (0, \infty)$. 
Note that $\Delta_K(x) \le - f_{K}(x) \le 0$ when $x\ge K-2$. 
It suffices to show $\Delta_K(x) \le 0$ for all $K>2$ and $x\in (0, K-2)$. 

For $K>2$ and $x\in (0, K-2)$, 
define 
\begin{align*}
    \tilde{\Delta}_K(x)
    & = \frac{2x}{K-2-x} \Delta_K(x) = F_K(x) - \frac{2x}{K-2-x} f_K(x). 
\end{align*}
It then suffices to show $\tilde{\Delta}_K(x) \le 0$ for all $K>2$ and $x\in (0, K-2)$. 
By some algebra and \eqref{eq:deriv_chi_density}, for $K>2$ and $x\in (0, K-2)$, 
\begin{align*}
    \frac{\deri}{\deri x}\tilde{\Delta}_K(x)
    & = 
    f_K(x) - \frac{2(K-2)}{(K-2-x)^2} f_K(x) - \frac{2x}{K-2-x} \dot{f}_K(x)
    =  - \frac{2(K-2)}{(K-2-x)^2} f_K(x) \le 0. 
\end{align*}
We can verify that $\lim_{x\rightarrow 0+} \tilde{\Delta}_K(x) = 0$ for any $K > 2$. 
Thus, we must have $\tilde{\Delta}_K(x) \le 0$ for all $K > 2$ and $x\in (0, K-2)$. 

From the above, Lemma \ref{lemma:diff_logdensity_L_epsilon} holds. 
\end{proof}

For a real function $\psi$ defined on $\mathcal{I} \subset \mathbb{R}$, 
define the number of sign changes of $\psi$ in $\mathcal{I}$ as 
\begin{align}\label{eq:sign_change}
    S^{-}_{\mathcal{I}}(\psi) = S^{-}_{\mathcal{I}}(\psi(x)) = \sup S^{-}_{\mathcal{I}} [\psi(x_1), \psi(x_2), \ldots, \psi(x_m)]
\end{align}
where $S^{-}_{\mathcal{I}}(y_1, y_2, \ldots, y_m)$ is the number of sign changes of the sequence $(y_1, y_2, \ldots, y_m)$ with the zero terms being discarded, and the supremum in \eqref{eq:sign_change} is over all sets $x_1 < x_2 < \cdots < x_m$ with $x_i\in \mathcal{I}$ and $m<\infty$. 
For any $c\in \mathbb{R}$ and function $\psi$, define $\psi_c(x) = f(x-c)$. 

\begin{lemma}\label{lemma:disperse_sign_change}
Let $F$ and $G$ be two absolutely continuous distributions having intervals as their support, in the sense that each of $F$ and $G$ has a probability density function that takes positive values on an interval and zero values otherwise, 
and let $f$ and $g$ be the corresponding densities. 
If $\mathcal{S}^{-}_{\mathbb{R}}(f_c-g)\le 2$ for all $c\in \mathbb{R}$, 
with the sign sequence being $-, +, -$ in case of equality, then $F$ is less dispersed than $G$. 
\end{lemma}

\begin{proof}[of Lemma \ref{lemma:disperse_sign_change}]
Lemma \ref{lemma:disperse_sign_change} follows from \citet[][Theorem 2.5]{shaked1982}. 
\end{proof}

\begin{lemma}\label{lemma:L_disperse}
For any $a \in [0, \infty]$ and integer $K\ge 1$, 
$L_{K,a}$ is less dispersed than $\varepsilon_0$. 
\end{lemma}

\begin{proof}[of Lemma \ref{lemma:L_disperse}]
Lemma \ref{lemma:L_disperse} holds obviously when $a$ equals zero or infinity. Below we consider only the case with $a\in (0, \infty)$.

Let $g_{K,a}$ and $\phi$ be the densities of $L_{K,a}$ and $\varepsilon_0$. We have derived the form of $g_{K,a}$ in Lemma \ref{lemma:Ldensity}. Furthermore, we define $g_{1, a}(x)$ to be zero when $x^2=a$; obviously, $g_{1, a}(x)$ is still the density of $L_{K,a}$.
Let $\mathcal{I} = (-\sqrt{a}, \sqrt{a})$ be the support of $L_{K,a}$. For any $c\in \mathbb{R}$, define $g_{K,a,c}(x) = g_{K,a}(x-c)$, $\mathcal{I}_c = (-\sqrt{a}+c, \sqrt{a}+c)$, and $h_{K,a,c}$ the same as in Lemma \ref{lemma:diff_logdensity_L_epsilon}. 
We then have, for any $x\in \mathcal{I}_c$, 
\begin{align*}
    \text{sign}\left\{ g_{K,a,c}(x) - \phi(x) \right\}
    & = \text{sign}\left\{ g_{K,a}(x-c) - \phi(x-c+c) \right\}
    = \text{sign}\left\{ h_{K,a,c}(x-c) \right\}. 
\end{align*}
Therefore, 
$S^{-}_{\mathcal{I}_c} (g_{K,a,c} - \phi) = S^{-}_{\mathcal{I}}(h_{K,a,c})$. 
By Lemma \ref{lemma:diff_logdensity_L_epsilon}, 
$h_{K,a,c}$ is a concave function on $\mathcal{I}$. 
This then implies that 
\begin{align*}
    S^{-}_{\mathcal{I}_c} (g_{K,a,c} - \phi) = S^{-}_{\mathcal{I}}(h_{K,a,c}) 
    = 
    \begin{cases}
    0, & \text{ with sign being $+$ or $-$}, \\
    1, & \text{ with sign sequence being $(-, +)$ or $(+, -)$}, \\
    2, & \text{ with sign sequence being $(-, +, -)$}. 
    \end{cases}
\end{align*}
Note that $g_{K,a,c}(x) = 0 < \phi(x)$ for $x\notin \mathcal{I}_c$. 
We can then verify that $S^{-}_{\mathbb{R}} (g_{K,a,c} - \phi)$ must have the following forms:
\begin{align*}
    S^{-}_{\mathbb{R}} (g_{K,a,c} - \phi) 
    & = 
    \begin{cases}
    0, & \text{ with sign being $-$}, \\
    2, & \text{ with sign sequence being $(-, +, -)$}. 
    \end{cases}
\end{align*}
By Lemma \ref{lemma:disperse_sign_change}, $L_{K,a}$ is less dispersed than $\varepsilon_0$. 
Therefore, Lemma \ref{lemma:L_disperse} holds. 
\end{proof}

\begin{lemma}\label{lemma:disp_sum}
Assume $X$ is less dispersed than $Y$. 
Let $W$ be a random variable independent of $X$ and $Y$. 
Let $f(w)$ be the density of $W$. 
If $f(w)>0$ and $\deri^2 \log f(w)/\deri w^2 \leq 0$ for all $w$, 
then $X+W$ is less dispersed than $Y+W$. 
\end{lemma}

\begin{proof}[of Lemma \ref{lemma:disp_sum}]
Lemma \ref{lemma:disp_sum} follows from \citet[][Theorem 7]{lewis_thompson_1981}.
\end{proof}

Equipped with the above lemmas, we now prove Theorem \ref{thm:disp_std_rem}.

\begin{proof}[of Theorem \ref{thm:disp_std_rem}]
Let $\varepsilon_1\sim \mathcal{N}(0,1)$ be independent of $(\varepsilon_0, L_{K,a})$. 
From Lemma \ref{lemma:L_disperse}, 
$L_{K,a}$ is less dispersed than $\varepsilon_1$, 
which immediately implies that $\rho L_{K,a}$ is less dispersed than $\rho \varepsilon_1$.  
Thus, Theorem \ref{thm:disp_std_rem} holds obviously when $\rho = 1$. 
Below we consider only the case where $0\le \rho <1$. 

By some algebra, the second derivative of the log-density of $\sqrt{1-\rho^2} \varepsilon_0 \sim \mathcal{N}(0, 1-\rho^2)$ is a constant $-(1-\rho^2)^{-1} < 0$. 
Thus, from Lemma \ref{lemma:disp_sum},
$\sqrt{1-\rho^2} \varepsilon_0 + \rho L_{K,a}$ is less dispersed than $\sqrt{1-\rho^2} \varepsilon_0 + \rho \varepsilon_1 \sim \mathcal{N}(0,1) \sim \varepsilon_0$.

From the above, Theorem \ref{thm:disp_std_rem} holds. 
\end{proof}

\section{Proof of Theorem \ref{thm:power_compare}} \label{s:powerCompare}

Let $\beta_{\textup{rr}}$ and $\beta_{\textup{cr}}$ denote the left and right hand sides of \eqref{eq:power_bound}, respectively.
We first consider the case with $\tilde{V} = V$, which implies $\tilde{R}^2 = R^2$. 
We can then verify that  
\begin{align*}
    \nu_{1-\osl}(R^2) - \nu_{1-\beta_{\textup{rr}}}(R^2)  = V^{-1/2}N^{1/2} \tau = z_{1-\osl} - z_{1-\beta_{\textup{cr}}}.  
\end{align*}
Assuming $\tau\ge 0$, we have $1-\osl \ge 1-\beta_{\textup{rr}}$. 
From Theorem \ref{thm:disp_std_rem}, 
$\nu_{1-\osl}(R^2) - \nu_{1-\beta_{\textup{rr}}}(R^2) \le z_{1-\osl} - z_{1-\beta_{\textup{rr}}}$. 
This then implies that 
$z_{1-\osl} - z_{1-\beta_{\textup{cr}}} \le z_{1-\osl} - z_{1-\beta_{\textup{rr}}}$. 
Consequently, we must have $\beta_{\textup{rr}} \ge \beta_{\textup{cr}}$, i.e., the inequality in \eqref{eq:power_bound} holds. 

We then consider the case where $\tau \ge \nu_{1-\osl} (\tilde{R}^2) \tilde{V}^{1/2} N^{-1/2}$.  
We can verify that 
\begin{align*}
    \nu_{1-\osl}(\tilde{R}^2) \tilde{V}^{1/2} - \nu_{1-\beta_{\textup{rr}}}(R^2) V^{1/2} 
    = 
    N^{1/2} \tau 
    = 
    z_{1-\osl} \tilde{V}^{1/2} - z_{1-\beta_{\textup{cr}}} V^{1/2}. 
\end{align*}
Assuming $\tau \ge \nu_{1-\osl} (\tilde{R}^2) \tilde{V}^{1/2} N^{-1/2}$, we have 
$\beta_{\textup{\textup{rr}}} \ge 1/2$. 
Note that the distribution $(1 - R^2)^{1/2} \epsilon_0 + R L_{K,a}$ is symmetric around zero.
From \citet[][Theorem 2]{li2018asymptotic},  
we can verify that 
$- \nu_{1-\beta_{\textup{\textup{rr}}}}(R^2) = \nu_{\beta_{\textup{rr}}}(R^2) \le z_{\beta_{\textup{rr}}} = - z_{1-\beta_{\textup{rr}}}$ 
and 
$
\nu_{1-\osl}(\tilde{R}^2) \le z_{1-\osl}. 
$
These imply that 
\begin{align*}
    z_{1-\osl} \tilde{V}^{1/2} - z_{1-\beta_{\textup{cr}}} V^{1/2}
    & = \nu_{1-\osl}(\tilde{R}^2) \tilde{V}^{1/2} - \nu_{1-\beta_{\textup{\textup{rr}}}}(R^2) V^{1/2}
    \le z_{1-\osl} \tilde{V}^{1/2} - z_{1-\beta_{\textup{\textup{\textup{rr}}}}} V^{1/2}.
\end{align*}
Consequently, we must have  $\beta_{\textup{\textup{rr}}} \ge \beta_{\textup{cr}}$, i.e., the inequality in \eqref{eq:power_bound} holds. 

From the above, Theorem \ref{thm:power_compare} holds.

\section{Proof of Theorem \ref{thm:sampleSizeRatio}} \label{s:proofSampSizeRatio}

Let $N_{\text{cr}}$ and $N_{\text{rr}}$ denote the sample sizes necessary to achieve power $\gamma$ under complete randomization and rerandomization, respectively, as provided by Theorem \ref{thm:sampSizeRerand}. We have:
\begin{align*}
	\frac{N_{\text{rr}}}{N_{\text{cr}}} = 
	    \left( 
	    \frac{\nu_{1-\osl}(\tilde{R}^2) \tilde{V}^{1/2} - \nu_{1 - \gamma}(R^2) V^{1/2} }{z_{1-\osl} \tilde{V}^{1/2} - z_{1 - \gamma} V^{1/2}} \right)^2. 
\end{align*}

We first consider the case with $\tilde{V} = V$. In this case, the ratio simplifies to 
\begin{align*}
    \frac{N_{\text{rr}}}{N_{\text{cr}}} = 
	    \left(
	    \frac{\nu_{1-\osl}(R^2) V^{1/2} - \nu_{1 - \gamma}(R^2) V^{1/2} }{z_{1-\osl} V^{1/2} - z_{1 - \gamma} V^{1/2}} \right)^2
	    = 
	    \left(
	    \frac{\nu_{1-\osl}(R^2) - \nu_{1 - \gamma}(R^2) }{z_{1-\osl} - z_{1 - \gamma}} \right)^2. 
\end{align*}
From Theorem \ref{thm:disp_std_rem}, when $\gamma\ge \osl$, 
we have 
$\nu_{1-\osl}(R^2) - \nu_{1 - \gamma}(R^2) \le z_{1-\osl} - z_{1 - \gamma}$. 
This immediately implies that $N_{\text{rr}}/N_{\text{cr}}\leq 1$. 

We then consider the case with $\gamma \ge 0.5$. 
From \citet[][Theorem 2]{li2018asymptotic}, both  $\nu_{1-\osl}(R^2)$ and $\nu_{\gamma}(R^2)$ are decreasing in $R^2$ and increasing in $K$ and $a$, 
and they are less than or equal to $z_{1-\osl}$ and $z_{\gamma}$, respectively. 
Consequently, 
\begin{align*}
    \nu_{1-\osl}(\tilde{R}^2) \tilde{V}^{1/2} - \nu_{1 - \gamma}(R^2) V^{1/2}
    & 
    = 
    \nu_{1-\osl}(\tilde{R}^2) \tilde{V}^{1/2} + \nu_{\gamma}(R^2) V^{1/2}
    \\
    & \le 
    z_{1-\osl} \tilde{V}^{1/2} + z_{\gamma} V^{1/2}\\
    & = z_{1-\osl} \tilde{V}^{1/2} - z_{1-\gamma} V^{1/2}, 
\end{align*}
which immediately implies that $N_{\text{rr}}/N_{\text{cr}} \leq 1$. 
Moreover, 
$\nu_{1-\osl}(\tilde{R}^2) \tilde{V}^{1/2} - \nu_{1 - \gamma}(R^2) V^{1/2}
= 
\nu_{1-\osl}(\tilde{R}^2) \tilde{V}^{1/2} + \nu_{\gamma}(R^2) V^{1/2}$ is decreasing in $R^2$ and increasing in $K$ and $a$. 
This then implies that $N_{\text{rr}}/N_{\text{cr}}$ is decreasing in $R^2$ and increasing in $K$ and $a$. 

From the above, Theorem \ref{thm:sampleSizeRatio} holds.

\section{Type-I Error Rates under Rerandomization and Complete Randomization}\label{sec:compare_actual_typeI}
From Theorem \ref{thm:powerRerand}, the type-I error rates of the $\alpha$-level tests under rerandomization and complete randomization are, respectively,  
\begin{align*}
	\alpha_{\textup{rr}} = \overline{\mathcal{V}}_{R^2}\left( \frac{\nu_{1-\osl} (\tilde{R}^2) \tilde{V}^{1/2}}{V^{1/2}} \right)
	\ \ \ \text{ and } \ \ \  
	\alpha_{\textup{cr}} = 
	\overline{\Phi}\left\{ \frac{z_{1-\osl} \tilde{V}^{1/2}}{V^{1/2} } \right\}.
\end{align*}
It is challenging to compare $\alpha_{\textup{rr}}$ and $\alpha_{\textup{cr}}$ theoretically, due to the possible difference between $R^2$ and $\tilde{R}^2$. Nevertheless, we conjecture that $\alpha_{\textup{rr}}\le \alpha_{\textup{cr}}$, due to the same reason discussed in \S \ref{sec:power_samplesize} of the main paper: We conservatively estimate the true distribution of $N^{1/2}(\hat{\tau} - \tau)$ by the same amount under both designs, and the true distribution under rerandomization is more concentrated around zero. See Appendix \ref{sec:numerical_power_rerand_less} for a numerical study.

Below we consider the limiting case with $a=0$, which can be a good approximation for rerandomization with a small threshold as suggested in \citet{morgan2012rerandomization}. 
When $a=0$, $\overline{\mathcal{V}}_{R^2}$ simplifies to the survival function of $\mathcal{N}(0, 1-R^2)$, 
and $\nu_{1-\osl} (\tilde{R}^2)$ simplifies to the $(1-\alpha)$-quantile of $\mathcal{N}(0, 1-\tilde{R}^2)$. 
Consequently, the type-I error rate under rerandomization simplifies to 
\begin{align}\label{eqn:power_rr_0}
	\alpha_{\textup{rr}}  =  \overline{\Phi}\left\{ \frac{z_{1-\osl}\tilde{V}^{1/2}(1-\tilde{R}^2)^{1/2}}{V^{1/2}(1-R^2)^{1/2}} \right\} = 
	\overline{\Phi}\left\{ \frac{z_{1-\osl}(\tilde{V}-V R^2)^{1/2}}{(V-VR^2)^{1/2}} \right\}, 
\end{align}
where the last equality holds because $VR^2 = \tilde{V} \tilde{R}^2$. 
Because 
\begin{align*}
	\frac{\tilde{V}-V R^2}{V-VR^2} - \frac{\tilde{V}}{V}
	& = 
	\frac{VR^2(\tilde{V}-V)}{V(V-VR^2)} \ge 0, 
\end{align*}
we have $\alpha_{\textup{rr}}  \le \alpha_{\textup{cr}}$. 
Moreover, if $R^2>0$ and $\tilde{V}> V$,  $\alpha_{\textup{rr}}$ is strictly less than $\alpha_{\textup{cr}}$. 

When rerandomization has a small threshold and the treatment effect has a small size, the power of rerandomization can be close to that in \eqref{eqn:power_rr_0}. 
This implies that the power under rerandomization can be smaller than that under complete randomization.

\section{Rerandomization Can Be Less Powerful Than Complete Randomization}\label{sec:numerical_power_rerand_less}

Theorem \ref{thm:power_compare} establishes that testing power is greater under rerandomization than complete randomization when $\tilde{V} = V$ or when the treatment effect $\tau \geq \nu_{1-\osl}(\tilde{R}^2) \tilde{V}^{1/2} N^{-1/2}$. Otherwise, rerandomization may exhibit less testing power than complete randomization, because inference under rerandomization can be more conservative than that under complete randomization, as discussed in  Appendix \ref{sec:compare_actual_typeI}. This additional conservativeness decreases power, but the additional precision from rerandomization increases power. To illustrate this trade-off, we consider a simple numerical example below. 

Suppose that $V = 1$ and $R^2 = 0.5$. We consider two cases, which correspond to $\tilde{V} = V$ and $\tilde{V} > V$, as in Theorem \ref{thm:power_compare}. In Case (i), the probability limits of our estimators are the same as the corresponding truth, i.e., $\tilde{V} = V$ and $\tilde{R}^2 = R^2$. Meanwhile, in Case (ii), inference is asymptotically conservative, in the sense that $\tilde{V}/V= 1.1 > 1$ and $\tilde{R}^2 = VR^2/\tilde{V} \approx 0.455$. Figure \ref{fig:power_example} shows the power (as in Theorem \ref{thm:powerRerand}) of the $0.05$-level one-sided test based on
the mean-difference estimator under complete randomization and rerandomization, with the scaled average treatment effect $\tau N^{1/2}/V^{1/2}$ ranging from $0$ to $0.5$. From Fig. \ref{fig:power_example}(a), when inference is not conservative, the power at $\tau=0$ equals the nominal level $0.05$ under both designs, and rerandomization provides better power than complete randomization. From Fig. \ref{fig:power_example}(b), when we can only conduct conservative inference, the power at $\tau=0$ is less than the nominal $0.05$ under both designs, and moreover, the test is more conservative under rerandomization. However, the power of rerandomization quickly passes that of complete randomization when $N^{1/2}\tau/V^{1/2}$ is not too small, and the cutoff for $N^{1/2}\tau/V^{1/2}$ is much smaller than the theoretical cutoff $\nu_{1-\osl} (\tilde{R}^2) \tilde{V}^{1/2}/V^{1/2}\approx 1.27$ in Theorem \ref{thm:power_compare}. 
In addition, in Fig. \ref{fig:power_example}(c) we also consider Case (iii), which is the same as Case (ii), except that our inference is much more conservative with $\tilde{V} = 10$. 
In this case, the power of rerandomization also passes that of complete randomization when $N^{1/2}\tau/V^{1/2}$ is not too small, 
and the cutoff becomes closer to the theoretical cutoff $\nu_{1-\osl} (\tilde{R}^2) \tilde{V}^{1/2}/V^{1/2}\approx 5.07$.

\begin{figure}
    \centering
    \begin{subfigure}{0.33\textwidth}
    \centering
        \includegraphics[width=\textwidth]{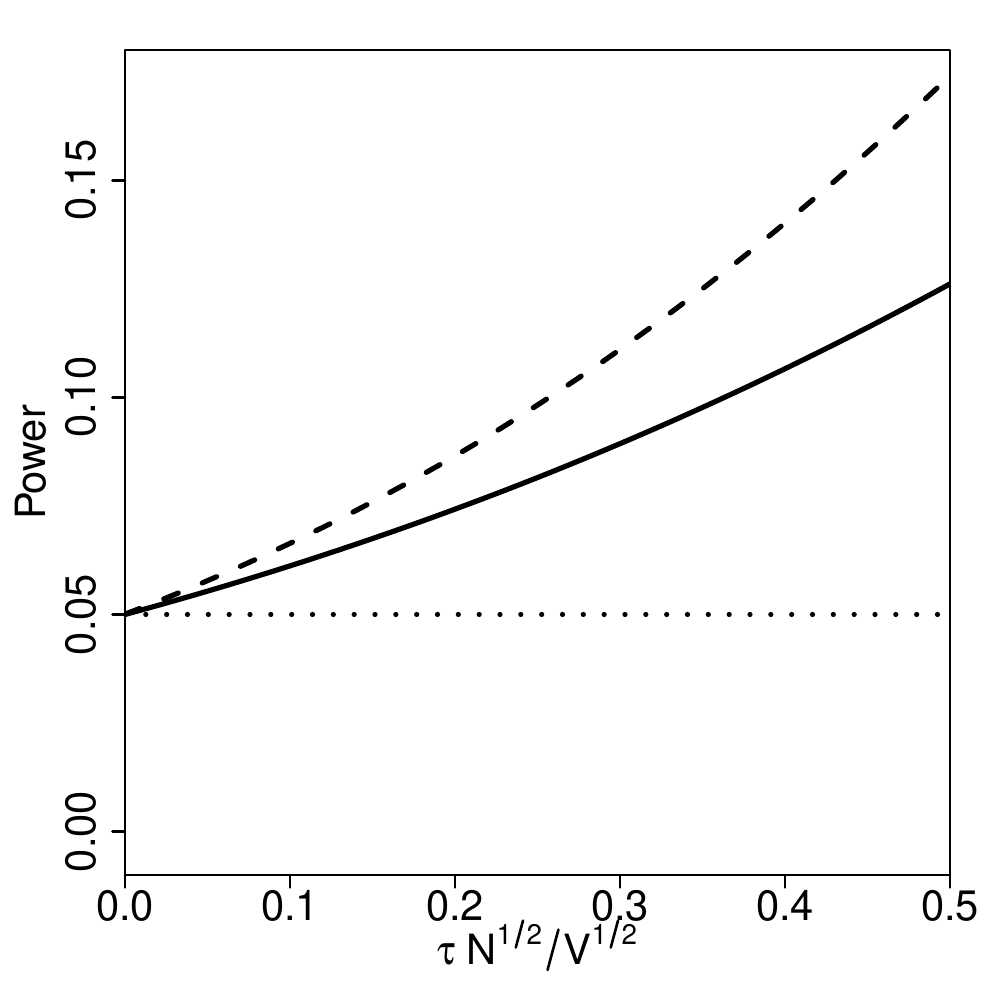}
        \caption{Case (i): $\tilde{V} = V$.}
    \end{subfigure}%
    \begin{subfigure}{0.33\textwidth}
    \centering
        \includegraphics[width=\textwidth]{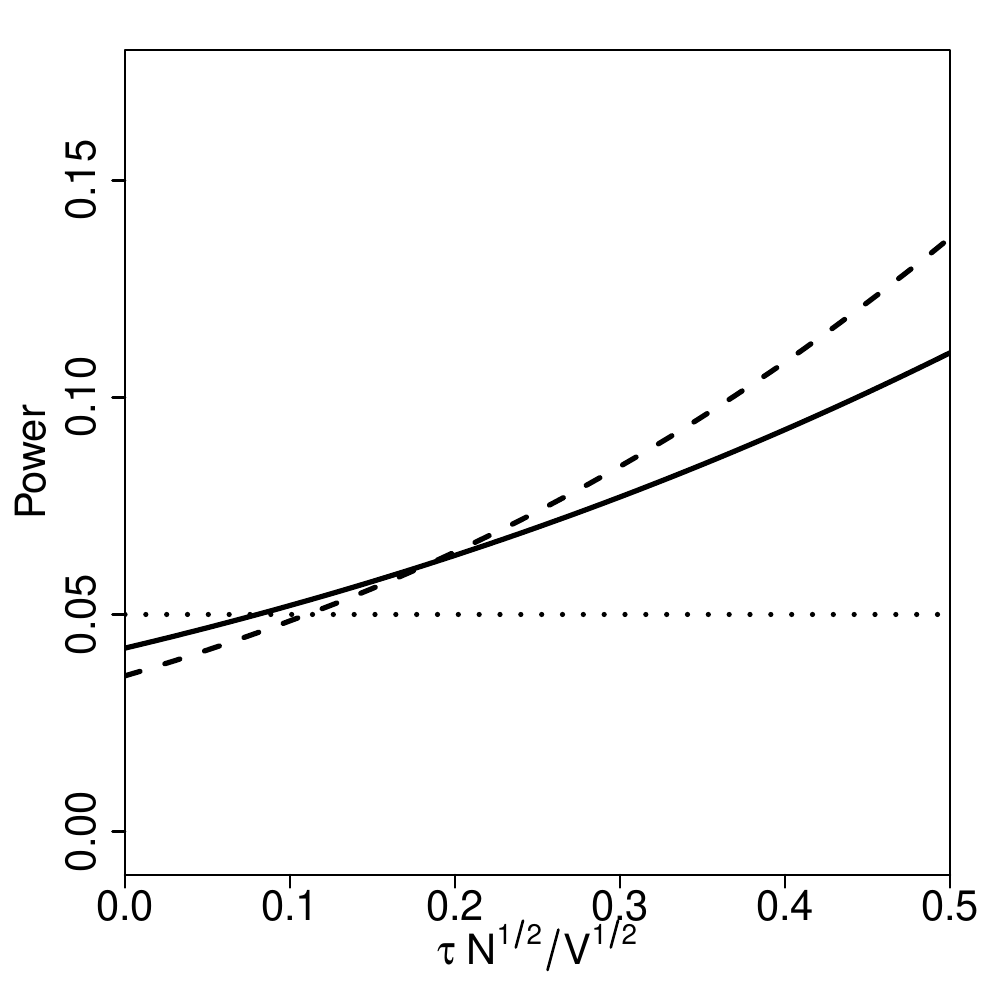}
        \caption{Case (ii): $\tilde{V} = 1.1 V$
        }
    \end{subfigure}%
    \begin{subfigure}{0.33\textwidth}
    \centering
        \includegraphics[width=\textwidth]{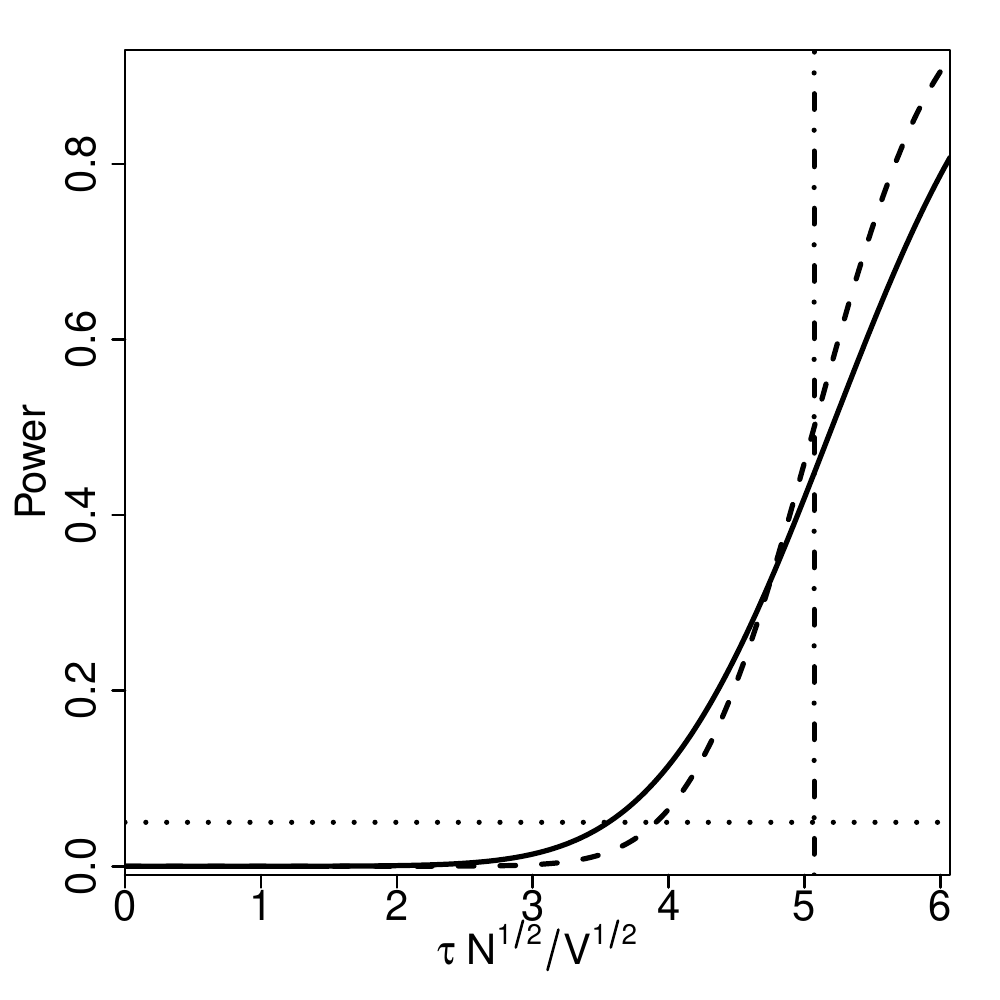}
        \caption{Case (iii): $\tilde{V} = 10 V$
        }
    \end{subfigure}
    \caption{
    Lower bound of the power for the 0.05-level two-sided test using the mean-difference estimator under complete randomization (solid line) and rerandomization (dashed line). 
    The dotted horizontal line denotes 0.025.
    The dotdash vertical line in (c) refers to the threshold $\nu_{1-\osl} (\tilde{R}^2) \tilde{V}^{1/2}$ as in Theorem \ref{thm:power_compare}. 
    }
    \label{fig:power_example}
\end{figure}

\section{Comparing Sample Size for Rerandomization and Complete Randomization} \label{s:simulations}

\subsection{Setup and Parameters}

Theorem \ref{thm:sampleSizeRatio} establishes, for any significance level $\alpha$ and power $\gamma$, the ratio between the sample size needed under rerandomization to achieve power $\gamma$ and the sample size needed under complete randomization:  
\begin{align}
    \frac{N_{\text{rr}}}{N_{\text{cr}}} = 
        \left\{ 
        \frac{\nu_{1-\osl}(\tilde{R}^2) \tilde{V}^{1/2} - \nu_{1 - \gamma}(R^2) V^{1/2} }{z_{1-\osl} \tilde{V}^{1/2} - z_{1 - \gamma} V^{1/2}} \right\}^2 \label{eqn:sampSizeRatio}
\end{align}
where $V$ denotes the variance defined in (\ref{eqn:compRandDistribution}), $\tilde{V}$ denotes the probability limit of its corresponding estimator defined in (\ref{eqn:V_tilde}), $\nu_{\alpha}(\rho^2)$ denotes the $\alpha$-quantile of the distribution $(1 - \rho^2)^{1/2} \epsilon_0 + \rho L_{K,a}$ in (\ref{eqn:rerandDistribution}), and $z_\alpha$ denotes the $\alpha$-quantile of the standard Normal distribution. This ratio depends on the number of covariates $K$, the correlation $R^2$, and the rerandomization threshold $a$. In this section, we present a simulation study to better understand how $N_{\text{rr}}/N_{\text{cr}}$ behaves for different $K$, $R^2$, and $a$, as well as varying levels of treatment effect heterogeneity. The smaller the ratio, the larger the benefits of rerandomization over complete randomization in terms of sample size.

We consider the dimension of covariates $K \in \{1, 10, 20, \dots, 100\}$, correlation $R^2 \in \{0, 0.1, \dots, 0.9\}$, and acceptance probabilities $p_a \in \{0.001, 0.01, 0.1\}$, where $p_a = \pr(M \leq a)$, i.e., the probability that a given randomization fulfills the rerandomization criterion. For simplicity, we focus on significance level $\alpha = 0.05$ and power $\gamma = 0.8$, both of which are common values in the power analysis literature. We found that results were consistent across other values of $\alpha$ and $\gamma$. The sample size ratio $N_{\text{rr}}/N_{\text{cr}}$ in Theorem \ref{thm:sampleSizeRatio} also depends on the non-Normal quantiles $\nu_{1-\osl}(\tilde{R}^2)$ and $\nu_{1-\gamma}(R^2)$, which in turn depend on $K$, $R^2$, and $p_a$. For each $K$, $R^2$, and $p_a$, we simulate $10^6$ draws from the non-Normal distributions $(1 - \tilde{R}^2)^{1/2} \epsilon_0 + \tilde{R} L_{K,a}$ and $(1 - R^2)^{1/2} \epsilon_0 + R L_{K,a}$, in order to approximate the quantiles $\nu_{1-\osl}(\tilde{R}^2)$ and $\nu_{1-\gamma}(R^2)$, respectively. Note that, when there is no treatment effect heterogeneity, $\tilde{R}^2 = R^2$; and when there is treatment effect heterogeneity, $\tilde{R}^2 = VR^2/\tilde{V}$.

We will first consider the case where there is no treatment effect heterogeneity, such that $S^2_{\tau} = 0$ and thus $\tilde{V} = V$. As a result, the sample size ratio does not depend on the potential outcome variances $S^2_1$ and $S^2_0$. Then we will consider the case where there is treatment effect heterogeneity, and thus $S^2_1$, $S^2_0$, and $S^2_{\tau}$ will affect the sample size ratio.

\subsection{Without Treatment Effect Heterogeneity} \label{ss:additivitySimulations}

Figure \ref{fig:sampleSizeRatios} displays $N_{\text{rr}}/N_{\text{cr}}$ for different combinations of $K$, $R^2$, and $p_a$. There are several observations from Fig. \ref{fig:sampleSizeRatios}, all of which validate the statements made in Theorem \ref{thm:sampleSizeRatio}. First, the ratio is always below 1. This confirms that there are always sample size benefits when running a rerandomized experiment, compared to a completely randomized experiment, at least when $\gamma \geq 0.5$. Furthermore, the ratio is decreasing in $R^2$ and increasing in $K$ and $p_a$. This demonstrates that the sample size benefits of rerandomization are large when a stringent criterion is used to balance a few covariates that are strongly related with experimental outcomes. More generally, Fig. \ref{fig:sampleSizeRatios} shows that rerandomization can lead to substantial sample size gains: For example, if $p_a = 0.001$, the median of the ratios in Fig. \ref{fig:sampleSizeRatios} is 0.75, and if further $R^2 \geq 0.3$ and $K \leq 50$, the median of the ratios is 0.58. This suggests that rerandomization can reduce sample size by 25\% to 40\%, compared to complete randomization.

\begin{figure}
\centering
    \includegraphics[scale=0.4]{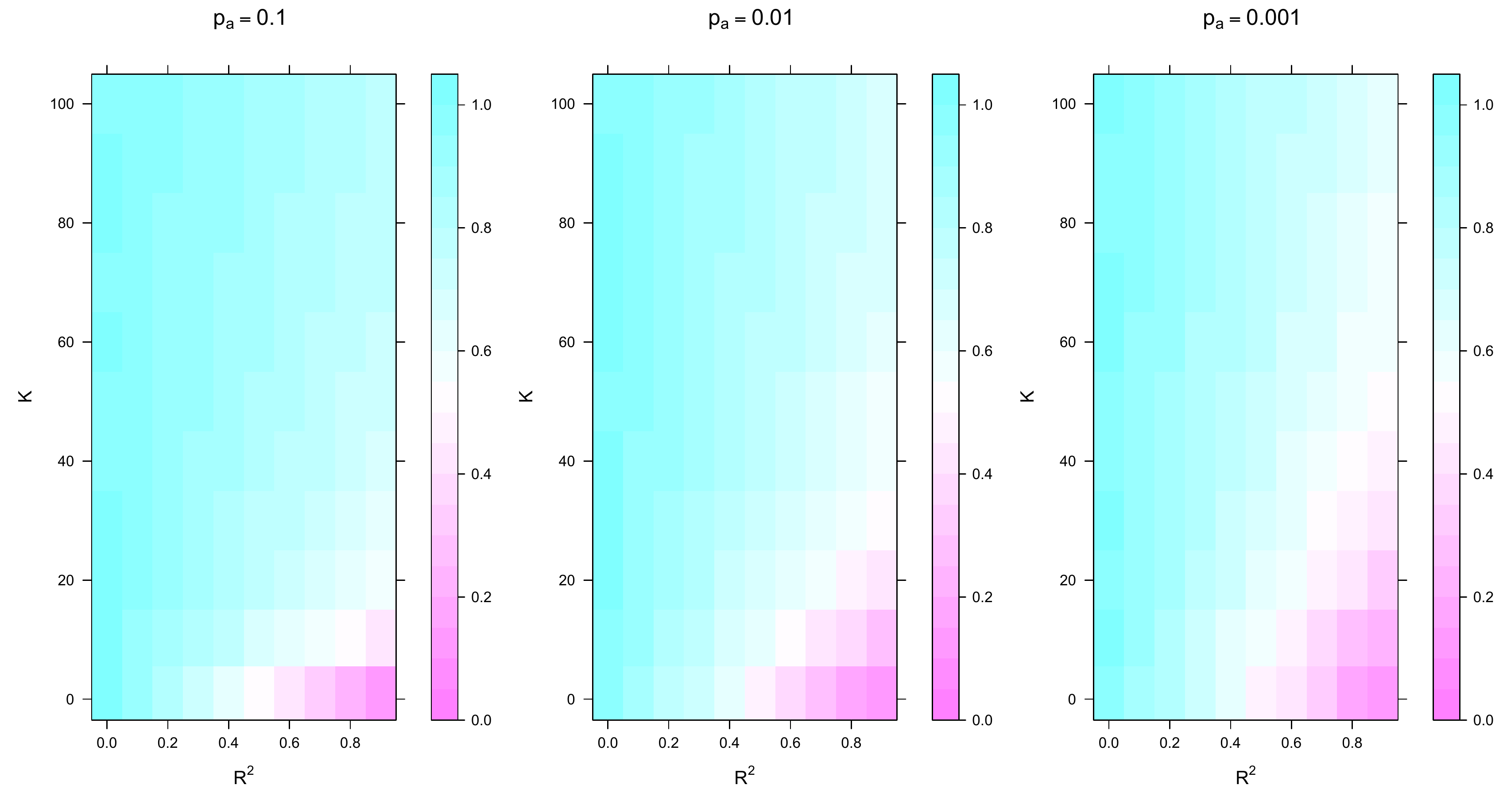}
    \caption{The ratio $N_{\text{rr}}/N_{\text{cr}}$ for different $K$, $R^2$, and $p_a$, when there is no treatment effect heterogeneity.}
    \label{fig:sampleSizeRatios}
\end{figure}

Because there is no treatment effect heterogeneity and thus $\tilde{V} = V$, the results in Fig. \ref{fig:sampleSizeRatios} hold for any degree of potential outcome variation $S^2_1$ and $S^2_0$ and any average treatment effect $\tau$, as shown by (\ref{eqn:sampSizeRatio}). According to Theorem \ref{thm:sampSizeRerand}, the sample size needed to achieve power $\gamma$ is increasing in $S^2_1, S^2_0$ and decreasing in $\tau$ for both rerandomization and complete randomization, which is a special case of rerandomization when $a = \infty$ or $R^2 = 0$. Thus, as $S^2_1$ and $S^2_0$ increase and as $\tau$ decreases, the nominal sample size gains from rerandomization can be arbitrarily large, at least as long as the desired power is greater than 50\%. For example, consider conducting an experiment where we desire 80\% power and $S_1 = S_0 = 4$. When $\tau = 2$, i.e. half a standard deviation, which is a medium effect according to a commonly used effect size rule-of-thumb by \cite{cohen2013statistical}, the necessary sample size under complete randomization is $N_{\text{cr}} \approx 99$. In this case, one may view the results in Fig. \ref{fig:sampleSizeRatios} as modest: If the covariates are modestly related to the outcomes ($R^2 = 0.3$), there are a moderate amount of covariates ($K = 50$), and we use a somewhat stringent rerandomization criterion ($p_a = 0.01$), we would expect only an approximately 13.3\% reduction in sample size under rerandomization, or approximately 14 fewer subjects. However, when we consider a small effect $\tau = 0.8$, or one-fifth of a standard deviation, $N_{\text{cr}} \approx 619$. In this scenario, a 13.3\% sample size reduction, or approximately 83 fewer subjects, may be considered quite large.

\subsection{With Treatment Effect Heterogeneity} \label{ss:heterogeneitySimulations}

Now we consider the case where $S_{\tau}^2 > 0$, and thus power and sample size will depend on the potential outcome variances $S^2_1$ and $S^2_0$ in addition to $S^2_{\tau}$. To our knowledge the literature has not discussed how treatment effect heterogeneity affects the power of completely randomized experiments, let alone rerandomized experiments. First we will discuss how treatment effect heterogeneity affects power and sample size for complete randomization and rerandomization, and then we will discuss how heterogeneity affects the sample size ratio $N_{\text{cr}}/N_{\text{rr}}$.

The asymptotic power for rerandomized experiments is characterized by Theorem \ref{thm:powerRerand}; for fixed values of $S^2_1$, $S^2_0$, and $\tau$, the power is increasing in $S^2_{\tau}$ as long as $\tau \geq \nu_{1-\osl}(\tilde{R}^2) \tilde{V}^{1/2} N^{-1/2}$; otherwise, it is decreasing in $S^2_{\tau}$. Because complete randomization is a special case of rerandomization, a similar result holds for completely randomized experiments, where power is increasing in $S^2_{\tau}$ as long as $\tau \geq z_{1-\osl} \tilde{V}^{1/2} N^{-1/2}$. Thus, treatment effect heterogeneity has a beneficial effect on power for relatively large effect sizes but an adverse effect for relatively small effect sizes. Furthermore, because $\nu_{1-\osl}(\tilde{R}^2) \leq z_{1-\osl}$ for all $\alpha \in (0, 0.5]$, power is increasing in $S^2_{\tau}$ for a wider range of effect sizes under rerandomization than under complete randomization. In other words, treatment effect heterogeneity is less likely to adversely affect power under rerandomization than under complete randomization. We illustrate this point further with numerical examples in Section \ref{s:numericalExamples}.

However, the results in the previous paragraph only hold when the variances $S^2_1$ and $S^2_0$ are fixed, and it's difficult to imagine a scenario where an increase in $S_\tau^2$ does not also increase $S^2_1$, which adversely affects power. For example, previous works studying treatment effect heterogeneity have considered data-generating models like $Y_i(1) = Y_i(0) + \tau + \sigma_{\tau} Y_i(0)$ for some heterogeneity parameter $\sigma_{\tau}$ \citep{ding2016randomization,branson2020sampling}. In this case, $S_1^2 = (1 + \sigma_{\tau})^2 S_0^2$ and $S^2_{\tau} = \sigma_{\tau}^2 S_0^2$, and thus more heterogeneity increases both $S^2_1$ and $S^2_{\tau}$. Because power tends to be decreasing in $S^2_1$ and $S^2_{\tau}$ for large $S^2_1$, this suggests that treatment effect heterogeneity generally has an adverse effect on power.

Meanwhile, from Theorem \ref{thm:sampSizeRerand}, the sample size necessary to achieve power $\gamma$ is decreasing in $S^2_\tau$ as long as $\gamma \geq 0.5$. Thus, for a fixed $\tau$ and power $\gamma \geq 0.5$, treatment effect heterogeneity has a beneficial effect on sample size for both completely randomized and rerandomized experiments. Indeed, this is analogous to the aforementioned results on power, because $\gamma \geq 0.5$ when $\tau \geq z_{1-\osl} \tilde{V}^{1/2} N^{-1/2}$ for completely randomized experiments and when $\tau \geq \nu_{1-\osl}(\tilde{R}^2) \tilde{V}^{1/2} N^{-1/2}$ for rerandomized experiments. However, if increased heterogeneity results in increased potential outcome variation, this may have an adverse affect on sample size, in the sense that increased $S^2_1$ will in turn increase sample size, as communicated in Theorem \ref{thm:sampSizeRerand}. These results are also illustrated further in Section \ref{s:numericalExamples}.

Finally, we consider how treatment effect heterogeneity affects the sample size ratio $N_{\text{rr}}/N_{\text{cr}}$. As communicated in Theorem \ref{thm:sampleSizeRatio}, when $S_{\tau} > 0$, the sample size ratio depends on two conservative estimators: $\tilde{V}$, which impacts inference for complete randomization and rerandomization, and $\tilde{R}^2 = V R^2 / \tilde{V}$, which only impacts inference for rerandomization. As a result, the sample size under rerandomization $N_{\text{rr}}$ is doubly-impacted by the conservative estimator $\tilde{V}$, thereby diminishing the sample size benefits of rerandomization when there is treatment effect heterogeneity. To demonstrate, let's consider an experiment where $p_1 = p_0 = 0.5$ and $S_1 = S_0 = 4$, the significance level is $\alpha = 0.05$, the desired power is $\gamma = 0.8$, and acceptance probability is $p_a = 0.001$. Figure \ref{fig:heteroSampSizeRatio} shows the resulting $N_{\text{rr}}/N_{\text{cr}}$ for treatment effect heterogeneity $S_{\tau} \in \{2, 4, 6\}$ for different $K$ and $R^2$. Many of the results from Fig. \ref{fig:sampleSizeRatios} still hold: $N_{\text{rr}}/N_{\text{cr}}$ is decreasing in $R^2$, increasing in $K$, and always below 1, as established by Theorem \ref{thm:sampleSizeRatio}. However, we see that this ratio is increasing in the treatment effect heterogeneity $S^2_{\tau}$; thus, rerandomization has less ability to reduce sample sizes when there is large treatment effect heterogeneity. When $S^2_\tau = 2$, $N_{\text{rr}}/N_{\text{cr}}$ is on average 2.4\% greater than when $S^2_{\tau} = 0$; when $S^2_{\tau} = 4$, $N_{\text{rr}}/N_{\text{cr}}$ is on average 9.9\% greater; and when $S^2_{\tau} = 6$, $N_{\text{rr}}/N_{\text{cr}}$ is on average 23.7\% greater. However, $S_\tau = 6$ denotes unusually large effect heterogeneity, because it is larger than $S_1$ and $S_0$. Furthermore, it's important to remember that the ratio result in Theorem \ref{thm:sampleSizeRatio} holds for any $\tau$; thus, as discussed in Section \ref{ss:additivitySimulations}, when $\tau$ is small, $N_{\text{cr}}$ will be large, making even small multiplicative sample size reductions possibly worthwhile.

\begin{figure}
    \centering
    \includegraphics[scale=0.4]{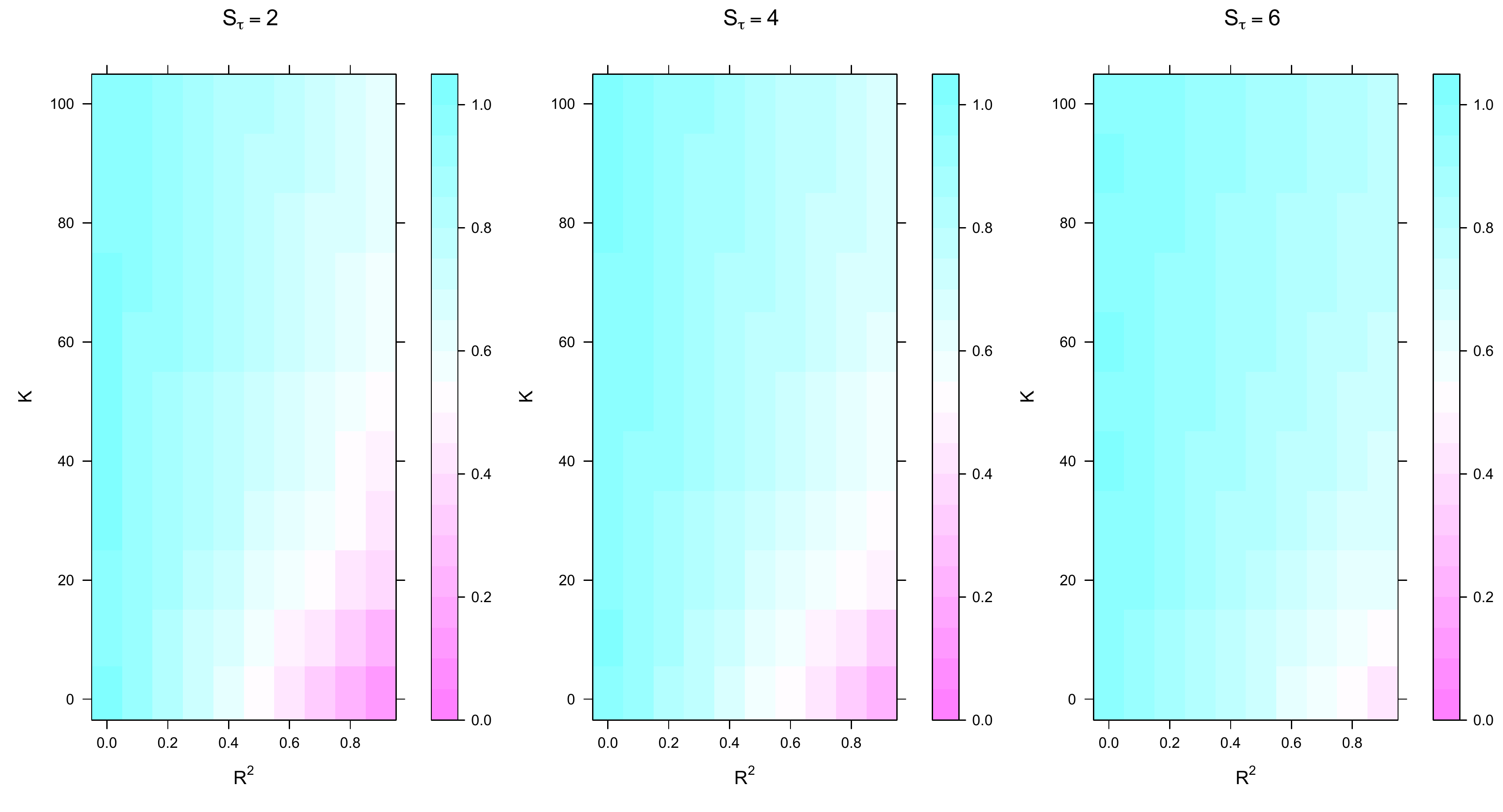}
    \caption{The sample size ratio $N_{\text{rr}}/N_{\text{cr}}$ when running an experiment with $p_1 = p_0 = 0.5$ and $S_1 = S_0 = 4$, where $\alpha = 0.05$, $\gamma = 0.8$, and $p_a = 0.001$. The three panels correspond to heterogeneity $S_{\tau} \in \{2, 4, 6\}$.}
    \label{fig:heteroSampSizeRatio}
\end{figure}

Furthermore, because the potential outcome variances $S_1^2$ and $S_0^2$ also impact power and sample size, they may also impact the ratio $N_{\text{rr}}/N_{\text{cr}}$. Let us consider the same example in Fig. \ref{fig:heteroSampSizeRatio}, but where we fix $K = 10$ and vary $S_1, S_0, S_{\tau}$, and $R^2$. Figure \ref{fig:heteroSampSizeRatio2} shows the ratio for different values of $S_1, S_0, S_{\tau}$, and $R^2$; in Fig. \ref{fig:heteroSampSizeRatio2} we restricted the color scale to $[0.25, 1.0]$ to more easily see trends for this plot. We see that as $S_1$ and $S_0$ increase, $N_{\text{rr}}/N_{\text{cr}}$ somewhat decreases, signaling that rerandomization can lead to larger sample size reductions when potential outcome variances are high. However, it appears that treatment effect heterogeneity has a relatively larger adverse impact on these sample size reductions; in other words, there is more variation with respect to the vertical axis in Fig. \ref{fig:heteroSampSizeRatio2} than the horizontal axis. Thus, if higher treatment effect heterogeneity in turn induces higher potential outcome variation, the adverse effects of heterogeneity will likely outweigh the beneficial effects of higher variation, thereby limiting the amount of sample size reductions we can expect from rerandomization.

\begin{figure}
    \centering
    \includegraphics[scale=0.4]{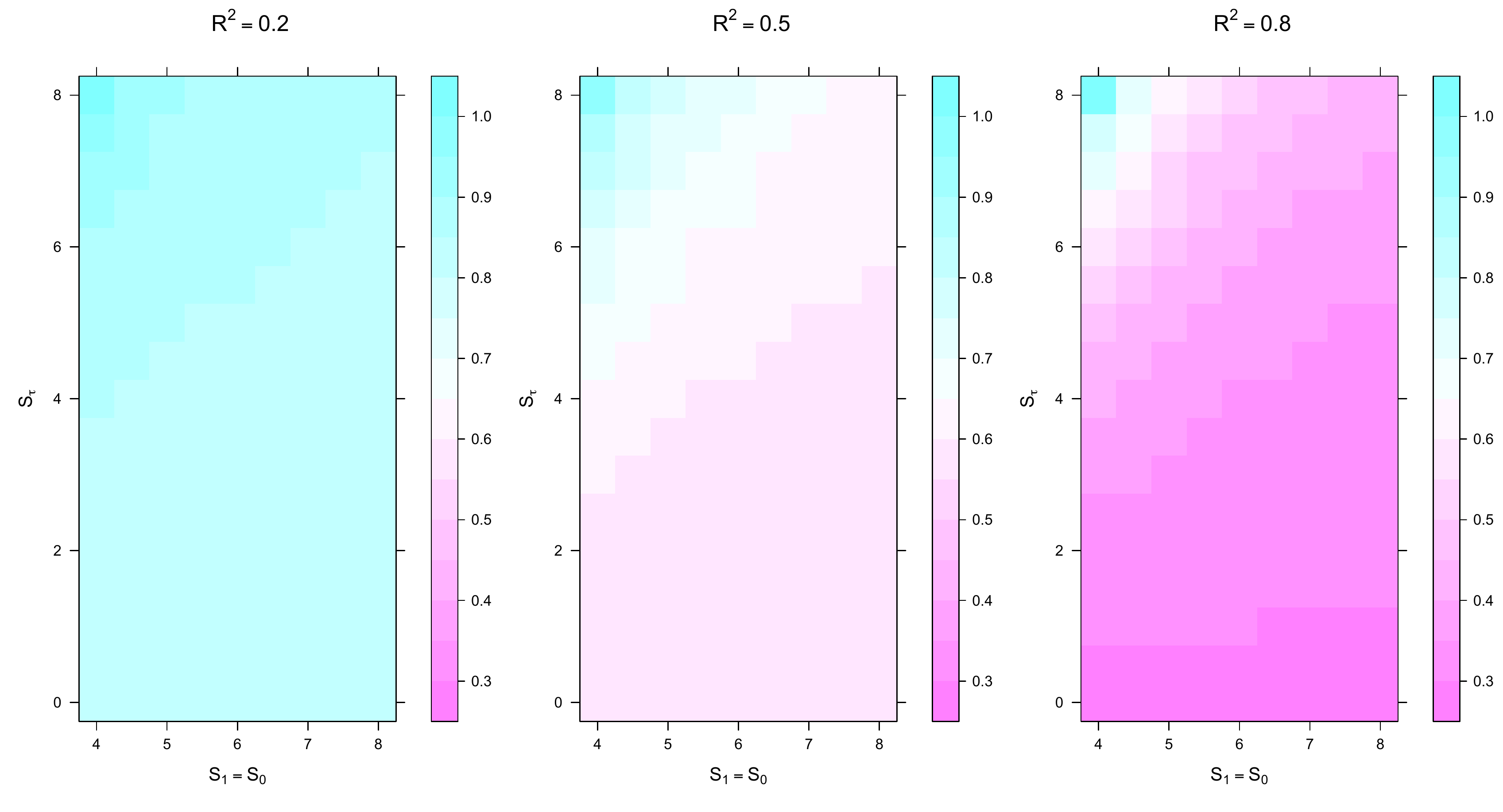}
    \caption{The sample size ratio $N_{\text{rr}}/N_{\text{cr}}$ when running an experiment with $p_1 = p_0 = 0.5$ and $S_1 = S_0 = 4$, where $\alpha = 0.05$, $\gamma = 0.8$, $K = 10$, and $p_a = 0.001$. The three panels correspond to $R^2 \in \{0.2, 0.5, 0.8\}$.
    }
    \label{fig:heteroSampSizeRatio2}
\end{figure}

\section{Additional Numerical Examples: How Treatment Effect Heterogeneity Affects Power and Sample Size in Randomized and Rerandomized Experiments} \label{s:numericalExamples}

In Section \ref{ss:heterogeneitySimulations}, we discussed how treatment effect heterogeneity affects power and sample size for completely randomized and rerandomized experiments according to Theorems \ref{thm:powerRerand} and \ref{thm:sampSizeRerand}. In this section, we present numerical examples to supplement that discussion. First we will present examples for completely randomized experiments, because, to our knowledge, the literature has not discussed how treatment effect heterogeneity affects power for completely randomized experiments, let alone rerandomized experiments. Then, we will discuss how these examples apply to rerandomized experiments.

As mentioned in Section \ref{ss:heterogeneitySimulations}, for fixed values of $S^2_1$, $S^2_0$, and $\tau$, testing power for completely randomized experiments is increasing in $S^2_{\tau}$ as long as $\tau \geq z_{1-\osl} \tilde{V}^{1/2} N^{-1/2}$, where $\tilde{V}$ is the probability limit of the estimator for $V$, defined in (\ref{eqn:compRandDistribution}). As a toy example, consider a completely randomized experiment where the proportions of treatment and control subjects are $p_1 = p_0 = 0.5$ and there are $N = 100$ subjects. Furthermore, say $S_1 = S_0 = 4$ and we use the ubiquitous variance estimator $\hat{V}_{\neyman} = p_1^{-1} s_1^2 + p_0^{-1}$ and thus $\tilde{V}_{\neyman} = p_1^{-1} S_1^2 + p_0^{-1} S_0^2$. Thus, power is increasing in $S^2_{\tau}$ as long as $\tau \geq z_{1-\osl} \cdot 8 \cdot 0.1 \approx 1.3$ for $\alpha = 0.05$. Figure \ref{fig:powerPlotSTau} shows power for this toy example when we vary $S_{\tau}$ for $\tau = 2$ and $\tau = 0.8$; we see that power is monotonically increasing in $S_{\tau}$ for the former but monotonically decreasing for the latter. This suggests that treatment effect heterogeneity has a beneficial effect on power for large effect sizes but an adverse effect for small effect sizes. Thus, if we incorrectly assume $S^2_{\tau} = 0$, which is common in power analyses, then we may underestimate power for large effect sizes but overestimate power for small effect sizes. Note that, in Fig. \ref{fig:powerPlotSTau}, $S_{\tau} = 8$ is very extreme; in this case, $V = 0$, and thus power is either 0 or 1, depending on whether $\tau \geq z_{1-\osl} \tilde{V}^{1/2} N^{-1/2}$.

Alternatively, we can consider a fixed average treatment effect $\tau$ and study power when we vary $S_1^2$ and $S_0^2$ in addition to $S_{\tau}^2$. Figures \ref{fig:powerPlotTau2} and \ref{fig:powerPlotTau0.8} show the power for the aforementioned toy example for $\tau = 2$ and $\tau = 0.8$, respectively, for different values of $S_1, S_0$ and $S_\tau$. When $\tau = 2$, power is monotonically increasing in $S_{\tau}$, as we saw in Fig. \ref{fig:powerPlotSTau}, but only for small values of $S_1$ and $S_0$; otherwise, it is monotonically decreasing. Meanwhile, we see that power is always monotonically decreasing in $S_{\tau}$ when $\tau = 0.8$. Furthermore, we see in Fig. \ref{fig:powerPlotTau2} and \ref{fig:powerPlotTau0.8} that power is monotonically decreasing in $S_1$ and $S_0$, which is already a well-known phenomenon in power analyses. Taking all of Fig. \ref{fig:powerPlots} together, treatment effect heterogeneity can have an adverse effect on power if $\tau$ is small or the potential outcome variances are large. It also appears that the potential outcome variances tend to have a more consequential effect on power than treatment effect variation. 

\begin{figure}
    \centering
    \begin{subfigure}[b]{0.32\textwidth}
     \includegraphics[width=\textwidth]{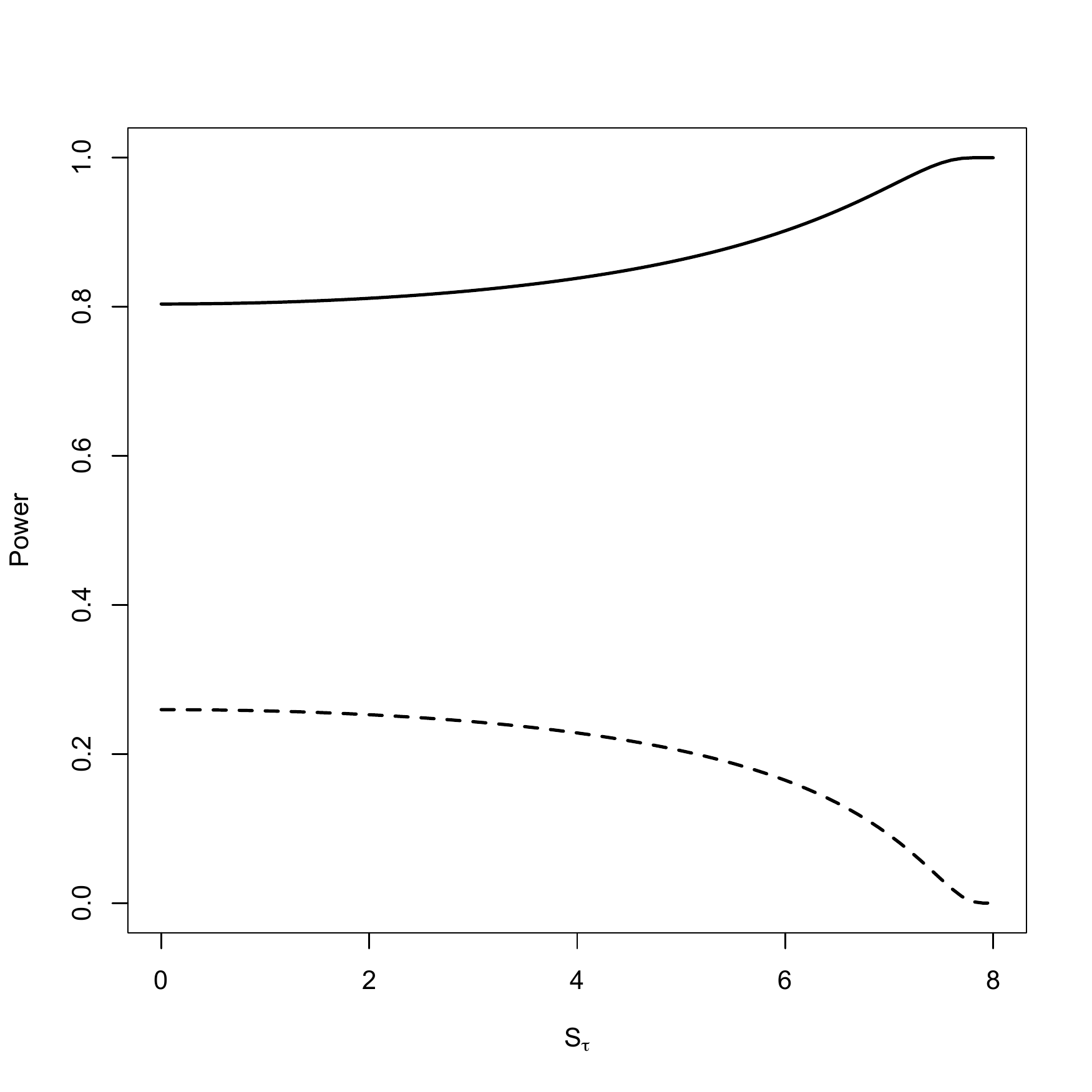}
     \caption{Power when varying $S_{\tau}$ for $S_1 = S_0 = 4$, when $\tau = 2$ (solid line) and $\tau = 0.8$ (dotted line).}
         \label{fig:powerPlotSTau}
    \end{subfigure}
    \hfill
     \begin{subfigure}[b]{0.32\textwidth}
         \centering
         \includegraphics[width=\textwidth]{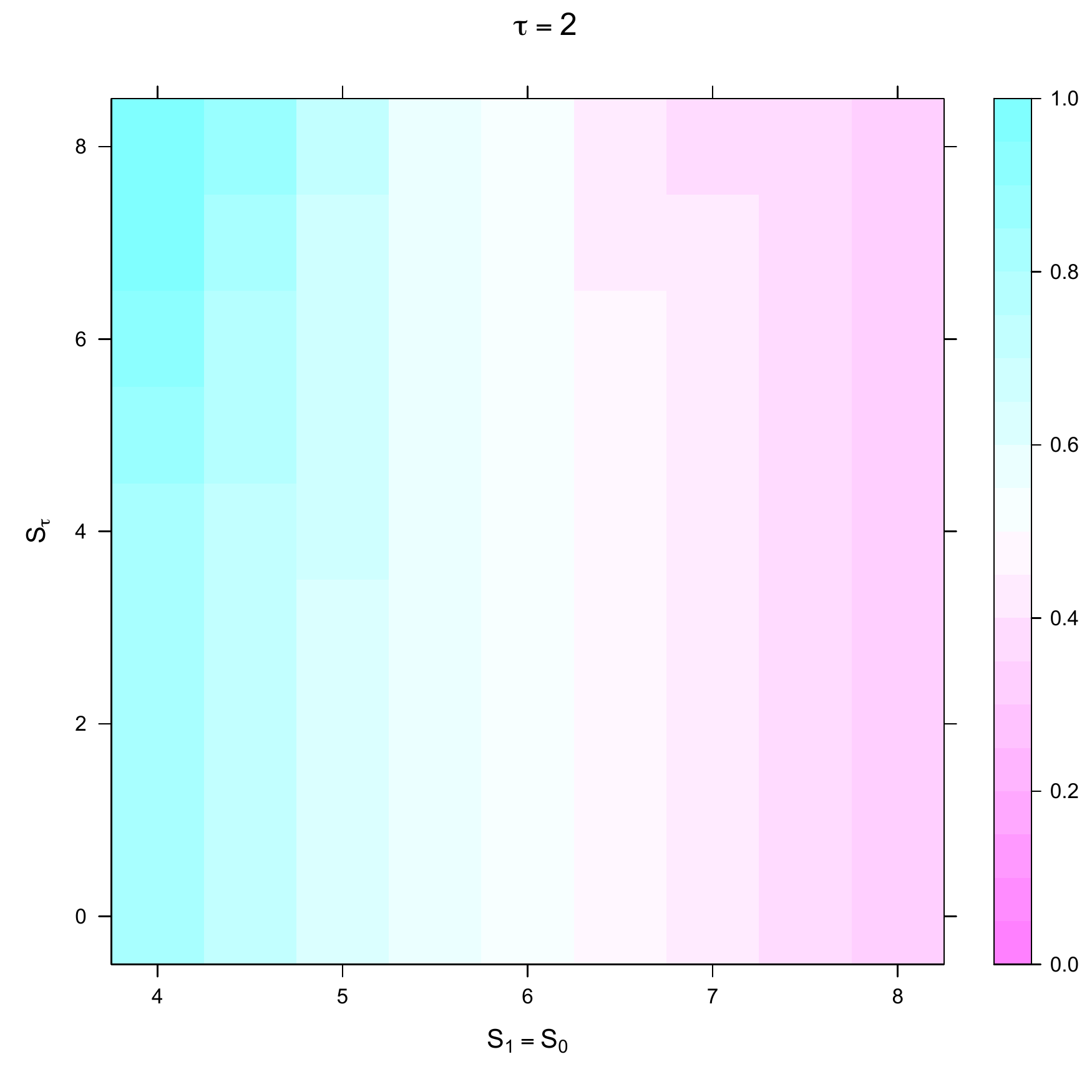}
         \caption{Power when varying $S_{\tau}$ as well as $S_1$ and $S_0$ for $\tau = 2$. Power ranges from 32.4\% to 100.0\%.}
         \label{fig:powerPlotTau2}
     \end{subfigure}
     \hfill
     \begin{subfigure}[b]{0.32\textwidth}
         \centering
         \includegraphics[width=\textwidth]{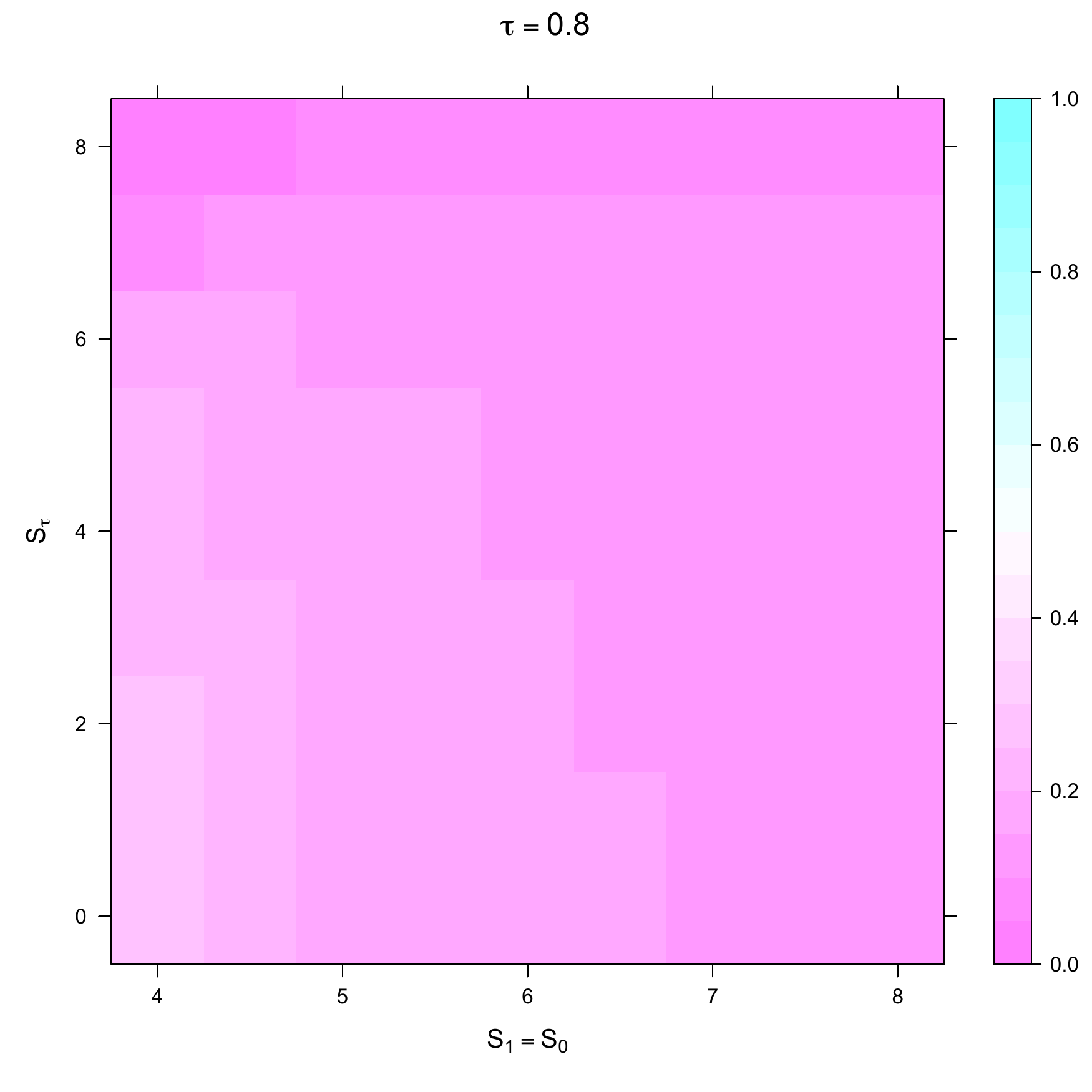}
         \caption{Power when varying $S_{\tau}$ as well as $S_1$ and $S_0$ for $\tau = 0.8$. Power ranges from 0.0\% to 26.0\%.}
         \label{fig:powerPlotTau0.8}
     \end{subfigure}
     \caption{Power of a completely randomized experiment under several scenarios when $p_1 = p_0 = 0.5$ and $N = 100$.}
     \label{fig:powerPlots}
\end{figure}

Now we'll consider sample size. When $\tau \geq z_{1-\osl} \tilde{V}^{1/2}N^{-1/2}$, power will always be greater than or equal to 50\%, as a consequence of Theorem \ref{thm:powerRerand}. Thus, for fixed $\tau$ and power $\gamma \geq 0.5$, treatment effect heterogeneity actually has a beneficial effect on the required sample size to achieve power $\gamma$, in the sense that larger $S^2_{\tau}$ leads to a smaller required sample size, as a consequence of Theorem \ref{thm:sampSizeRerand}. To demonstrate, let us again consider our toy example where $p_1 = p_0 = 0.5$, $S_1 = S_0 = 4$, and $\tau = 2$. Figure \ref{fig:sampSizePlotSTau} displays the sample size $N_{\text{cr}}$ to achieve power $\gamma = 0.8$ and $\gamma = 0.4$ under complete randomization for increasing values of $S_{\tau}$. When $\gamma = 0.8$, sample size is increasing in $S_{\tau}$, but it is decreasing in $S_{\tau}$ when $\gamma = 0.4$; this is analogous to our previous finding that power is increasing in $S_{\tau}$ only for treatment effects above a certain magnitude. Furthermore, note that in the extreme case when $S_{\tau} = 8$, $V = 0$, and thus $N_{\text{cr}}$ is no longer a function of $\gamma$.

Again we can also consider varying $S_1^2$ and $S_0^2$, in addition to $S^2_{\tau}$; the resulting sample size $N_{\text{cr}}$ required to achieve power $\gamma = 0.8$ is shown in Figures \ref{fig:sampSizePlotTau2} and \ref{fig:sampSizePlotTau0.8} for $\tau = 2$ and $\tau = 0.8$, respectively. For both of these scenarios, the sample size is decreasing in $S^2_{\tau}$, again suggesting that larger treatment effect heterogeneity can have a beneficial effect on the sample size $N_{\text{cr}}$ if $\gamma \geq 0.5$. However, in Fig. \ref{fig:sampSizePlotTau2} and \ref{fig:sampSizePlotTau0.8}, we see that there is more variation in $N_{\text{cr}}$ across the horizontal axis than the vertical axis. This suggests that potential outcome variances have a larger effect on sample size than treatment effect heterogeneity, validating common power analyses that focus on these quantities rather than treatment effect heterogeneity. In particular, if increased treatment effect heterogeneity in turn increases potential outcome variances, then in general heterogeneity may have an adverse effect on sample size.

\begin{figure}
    \centering
    \begin{subfigure}[b]{0.32\textwidth}
     \includegraphics[width=\textwidth]{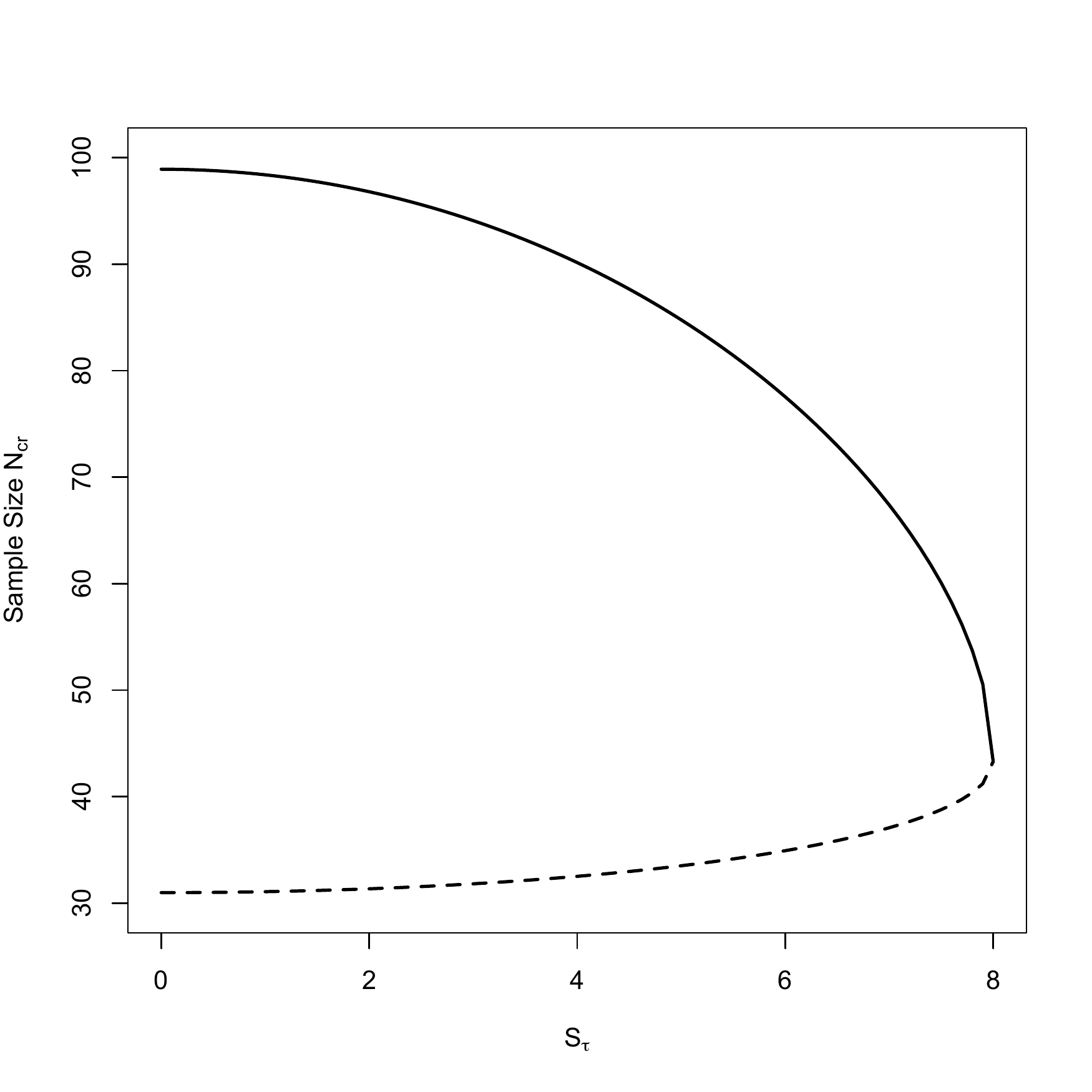}
     \caption{$N_{\text{cr}}$ when varying $S_{\tau}$ for $S_1 = S_0 = 4$ and $\tau = 2$, when $\gamma = 0.8$ (solid line) and $\gamma = 0.4$ (dotted line).}
         \label{fig:sampSizePlotSTau}
    \end{subfigure}
    \hfill
     \begin{subfigure}[b]{0.32\textwidth}
         \centering
         \includegraphics[width=\textwidth]{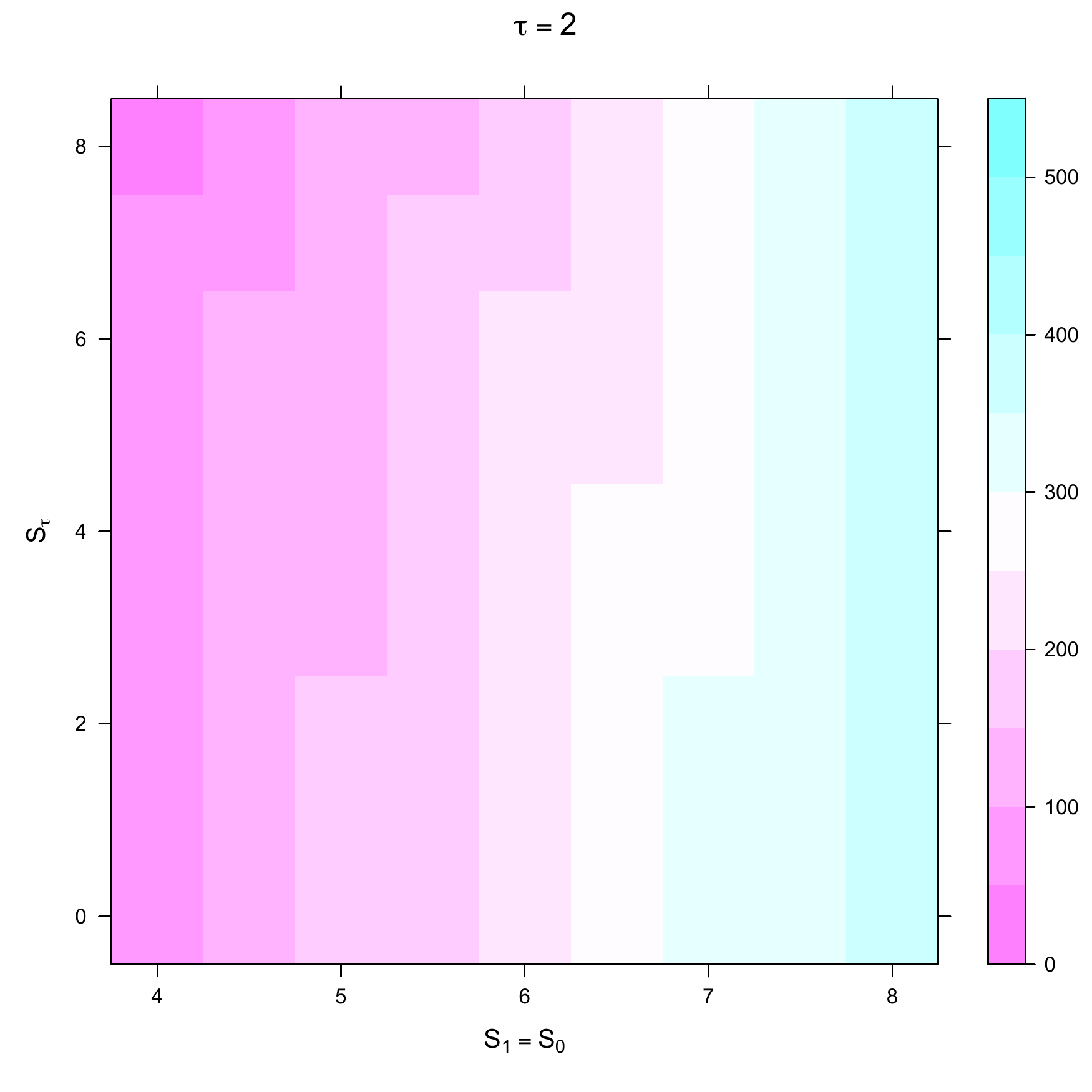}
         \caption{$N_{\text{cr}}$ when varying $S_{\tau}, S_1, S_0$ for $\tau = 2$. $N_{\text{cr}}$ ranges from approximately 44 to approximately 396.}
         \label{fig:sampSizePlotTau2}
     \end{subfigure}
     \hfill
     \begin{subfigure}[b]{0.32\textwidth}
         \centering
         \includegraphics[width=\textwidth]{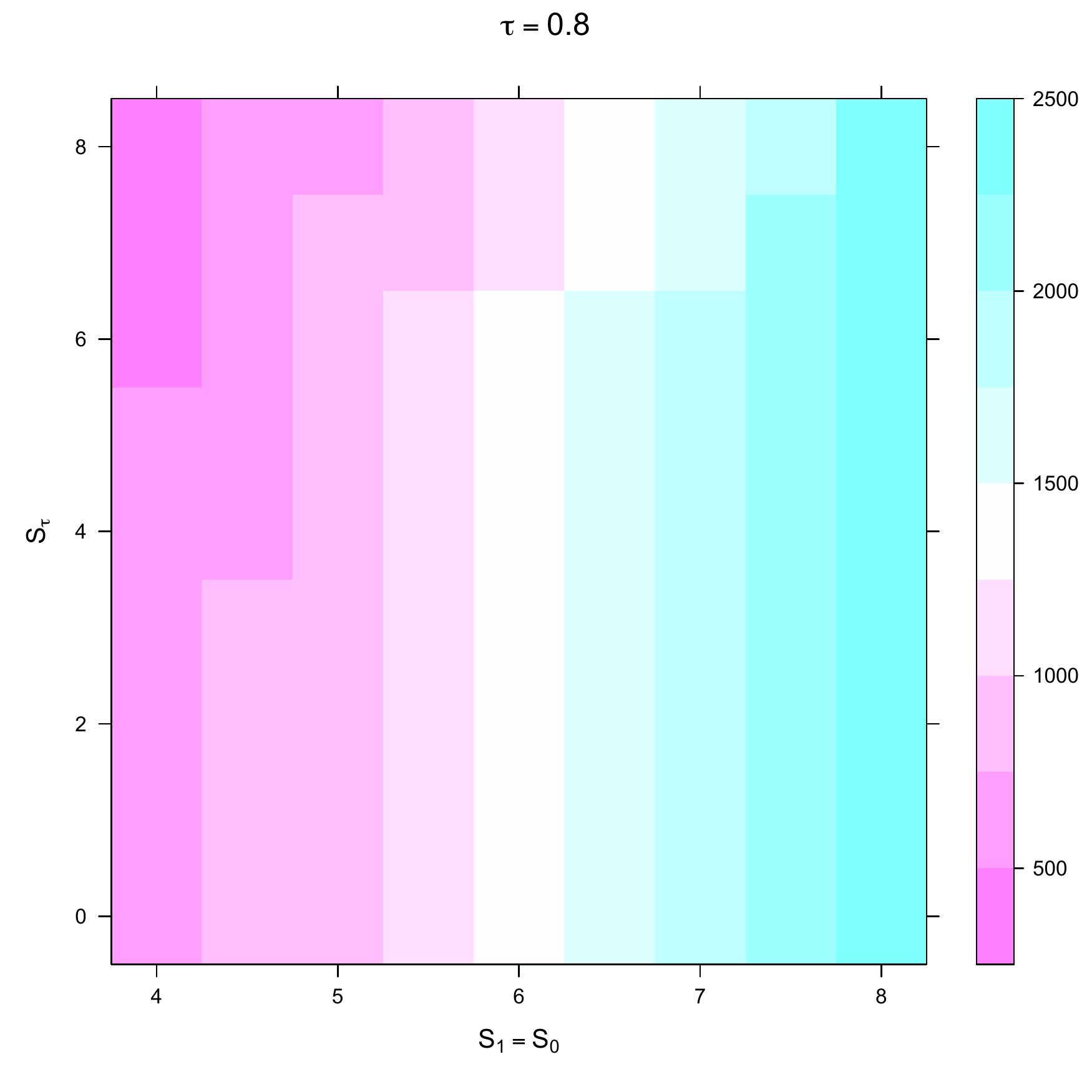}
         \caption{$N_{\text{cr}}$ when varying $S_{\tau}, S_1, S_0$ for $\tau = 0.8$. $N_{\text{cr}}$ ranges from approximately 271 to approximately 2473.}
         \label{fig:sampSizePlotTau0.8}
     \end{subfigure}
     \caption{Sample size $N_{\text{cr}}$ required to achieve power $\gamma$ when running a completely randomized experiment with $p_1 = p_0 = 0.5$ under different scenarios. In (b) and (c), $\gamma = 0.8$.}
     \label{fig:sampSizePlots}
\end{figure}

Finally, we can also consider how treatment effect heterogeneity affects power and sample size for rerandomized experiments. According to Theorem \ref{thm:powerRerand}, for fixed $S^2_1$, $S^2_0$, and $\tau$, power under rerandomization is increasing in $S^2_{\tau}$ as long as $\tau \geq \nu_{1-\osl}(\tilde{R}^2) \tilde{V}^{1/2} N^{-1/2}$, where $\nu_{1-\osl}(\tilde{R}^2)$ denotes the $(1-\osl)$-quantile of the distribution $\sqrt{1 - \tilde{R}^2} \epsilon_0 + \sqrt{\tilde{R}^2} L_{K,a}$ and $\tilde{R}^2 = VR^2 / \tilde{V}$. Note that $\nu_{1-\osl}(\tilde{R}^2) \leq z_{1-\osl}$ for all $\alpha \in (0,0.5)$, with equality only if $\tilde{R}^2 = 0$. Thus, the same conclusions made in this section for completely randomized experiments also hold for rerandomized experiments, but for smaller effect sizes. In other words, under rerandomization, a smaller $\tau$ is required in order for power to be increasing in $S^2_{\tau}$; or conversely, a smaller sample size $N_{\text{rr}}$ is required to achieve a certain level of power $\gamma \geq 0.5$, as established by Theorem \ref{thm:sampleSizeRatio}.

\section{Power and Sample Size Calculations using the \texttt{R} package \texttt{rerandPower}} \label{s:package} 

Practitioners may be interested in implementing power and sample size calculations for completely randomized and rerandomized experiments based on the results presented in this paper. Our \texttt{R} package \texttt{rerandPower}, available on \texttt{CRAN}, has four functions: \texttt{power.rand()}, \texttt{power.rerand()}, \texttt{sampleSize.rand()}, and \texttt{sampleSize.rerand()}.

The functions \texttt{power.rand()} and \texttt{power.rerand()} compute power for given sample sizes $N_1$ and $N_0$, potential outcome standard deviations $S_1$ and $S_0$, treatment effect heterogeneity standard deviation $S_\tau$, and average treatment effect $\tau$. The significance level $\alpha$ can also be specified. Let's consider the toy example in Section \ref{s:numericalExamples}, where $N_1 = N_0 = 50$, $S_1 = S_0 = 4$, $\tau = 2$, and $S_{\tau} = 0$ or $S_{\tau} = 4$ for a completely randomized experiment. The following lines of code implement the power calculations for these two cases, presented in Figure \ref{fig:powerPlots}:

\begin{verbatim}
> power.rand(N1 = 50, N0 = 50, s1 = 4, s0 = 4, tau = 2)
[1] 0.8037649
> power.rand(N1 = 50, N0 = 50, s1 = 4, s0 = 4, s.tau = 4, tau = 2)
[1] 0.838286
\end{verbatim}
We see that power increases when $S_\tau > 0$ because $\tau \geq z_{1-\osl} \tilde{V}^{1/2} N^{-1/2}$, as discussed in Section \ref{s:numericalExamples}. The calculation made in the first line of code, which by default sets $S_\tau = 0$, is widely available in other power analysis software; however, to our knowledge, other available software does not allow one to specify $S_\tau > 0$, as done in the second line of code.

Similar calculations can be made for rerandomized experiments using \texttt{power.rerand()}, except one also has to specify the number of covariates $K$, the correlation between covariates and potential outcomes $R^2$, and the acceptance probability $p_a = \pr(M \leq a)$, where $M$ is the Mahalanobis distance defined in (\ref{eqn:md}). When designing an experiment in practice, one can control $K$ and $p_a$, but of course one will not have knowledge about $R^2$ until the experiment has been conducted. Thus, $R^2$ should be specified based on subject-matter knowledge, or based on best- and worst-case scenarios. For example, consider conducting a rerandomized experiment where there are $K = 10$ covariates, $p_a = 0.01$, and there is a moderate correlation of $R^2 = 0.3$. Then the power under rerandomization for the same toy example above, with $S_\tau = 4$, is:
\begin{verbatim}
> power.rerand(N1 = 50, N0 = 50, s1 = 4, s0 = 4, s.tau = 4, tau = 2,
               K = 10, pa = 0.01, R2 = 0.3)
[1] 0.901424
\end{verbatim}
We see that, compared to the complete randomization example, power is higher in this case.

Meanwhile, the functions \texttt{sampleSize.rand()} and \texttt{sampleSize.rerand()} compute the sample size $N$ necessary to achieve a prespecified level of power $\gamma$ for given $S_1,S_0,S_\tau,\tau$, and sample size proportions $p_1 = N_1/N$ and $p_0 = N_0/N$. Let's again consider the toy example from Section \ref{s:numericalExamples}, where $p_1 = p_0 = 0.5$, $S_1 = S_0 = 4$, $\tau = 2$, and $S_\tau = 0$ or $S_\tau = 4$. The following lines of code implement sample size calculations for these two cases when power $\gamma = 0.8$, presented in Figure \ref{fig:sampSizePlots}:
\begin{verbatim}
> sampleSize.rand(power = 0.8, s1 = 4, s0 = 4, tau = 2)
[1] 98.92092
> sampleSize.rand(power = 0.8, s1 = 4, s0 = 4, s.tau = 4, tau = 2)
[1] 90.15267
\end{verbatim}
We see that the necessary sample size decreases when there is treatment effect heterogeneity, because $\gamma \geq 0.5$. Again, the calculation made in the first line of code is also widely available in other power analysis software, but to our knowledge, that in the second line of code is not.

We can again make similar calculations for a rerandomized experiment. Let's again consider the case where $K = 10$, $p_a = 0.01$, and $R^2 = 0.3$. Then, the sample size calculation for the case where $S_\tau = 4$ is:
\begin{verbatim}
> sampleSize.rerand(power = 0.8, s1 = 4, s0 = 4, s.tau = 4, tau = 2,
                    K = 10, pa = 0.01, R2 = 0.3)
[1] 72.6096
\end{verbatim}
We see that, for this example, rerandomization requires a smaller sample size than complete randomization, which will always be the case when $\gamma \geq 0.5$, as established by Theorem \ref{thm:sampleSizeRatio}. For this example, rerandomization decreases the sample size requirement by approximately 19.5\%.

More details and examples are available in the documentation for \texttt{rerandPower} on \texttt{CRAN}.

\bibliography{rerandPowerBib}

\bibliographystyle{biometrika}

\end{document}